\documentclass[a4paper]{article}
\usepackage{a4wide}
\usepackage[UKenglish]{babel}
\usepackage{graphicx,subfigure,caption}
\usepackage{algorithm,algorithmic}
\graphicspath{{./}{images/}}
\usepackage[colorlinks=true,linkcolor=blue,citecolor=red]{hyperref}
\usepackage{xcolor}
\usepackage{amssymb,amsthm,empheq,bbold}
\usepackage[inline]{enumitem}
\usepackage{calrsfs}\DeclareMathAlphabet{\pazocal}{OMS}{zplm}{m}{n}
\usepackage{authblk}

\DeclarePairedDelimiter\abs{\lvert}{\rvert}
\newcommand{\C}{\mathbb{C}}
\newcommand{\Lip}{\operatorname{Lip}}
\newcommand{\cM}{\mathfrak{M}}
\newcommand{\N}{\mathbb{N}}
\newcommand{\normHm}[2]{\Vert#1\Vert_{H^{#2}(\R_+)}}
\newcommand{\norminf}[1]{\Vert#1\Vert_\infty}
\newcommand{\cP}{\mathcal{P}}
\newcommand{\Prob}{\mathbb{P}}

\newcommand{\R}{\mathbb{R}}
\newcommand{\cS}{\mathfrak{S}}
\newcommand{\supp}[1]{\operatorname{supp}#1}
\newcommand{\cU}{\pazocal{U}}
\newcommand{\pX}{\pazocal{X}}

\newtheorem{assumption}{Assumption}[section]
\newtheorem{proposition}[assumption]{Proposition}
\newtheorem{theorem}[assumption]{Theorem}
\newtheorem{lemma}[assumption]{Lemma}
\newtheorem{corollary}[assumption]{Corollary}
\theoremstyle{remark}\newtheorem{remark}[assumption]{Remark}
\allowdisplaybreaks

\title{Linear inelastic kinetic equations modelling the spread of fake news and its interplay with personal awareness}
\author{Martina Fraia}
\author{Nadia Loy}
\author{Andrea Tosin}
\affil{{\small Department of Mathematical Sciences ``G. L. Lagrange'', Politecnico di Torino, Italy}}
\date{}
		
\begin{document}
\maketitle
	
\begin{abstract}
In this paper, we introduce a kinetic model which describes a learning process leading individuals to build personal awareness about fake news. Next, we embed the results of this model into another kinetic model, which describes the popularity gained by news on social media conditioned to the reliability of the disseminated information. Both models are formulated in terms of linear inelastic Boltzmann-type equations, of which we investigate the main analytical properties -- existence and uniqueness of solutions, trend to equilibrium, identification of the equilibrium distributions -- by employing extensively Fourier methods for kinetic equations. We also provide evidence of the analytical results by means of Monte Carlo numerical simulations.

\medskip


\medskip

\end{abstract}
	
\section{Introduction}

In an age dominated by digital communication and social media platforms, the spread of fake news has become a challenging social issue. False information, often designed to manipulate emotions or reinforce biases, can spread rapidly, influencing public opinion, behaviour, and even decision-making processes. Although much research has focussed on identifying and fighting fake news, less attention has been paid to understanding how individual awareness and cognitive abilities evolve after exposure to such contents. This paper aims to model the impact of fake news on personal awareness, particularly in terms of the ability of individuals to discern credible information from misinformation and to avoid spreading misinformation further.

In the literature, some mathematical models have been proposed to investigate the spread of fake news taking inspiration from the spread of infectious diseases. This pioneering idea appeared, probably for the first time, in~\cite{daley1964NATURE}; then, in recent years, it has been rediscovered thanks to the popularity gained by epidemiological models, see e.g.,~\cite{castiello2023ALGORITHMS,piqueira2020PHYSA}. Those models describe the diffusion of fake news like the contagion of an infectious disease, partitioning the society in a certain number of compartments such as: the oblivious individuals, i.e., those who are unaware of fake news; the individuals exposed to fake news; the fake news spreaders, i.e., the analogous of the infectious individuals; and the individuals who refrain from spreading fake news, virtually the analogous of the recovered individuals in an epidemiological context. A further elaboration has consisted in introducing the individual \textit{competence} as a variable structuring the models, cf. e.g.,~\cite{franceschi2022PTRSA,franceschi2022PDEA}, building on the kinetic description of personal conviction/knowledge originally introduced in~\cite{brugna2015PRE,pareschi2014PTRSA}. Typically, the competence evolves through interactions among the individuals or with a background and is assumed to affect the switch rates of the individuals across the compartments. On a closely related topic, we also mention models investigating the social impact of opinion formation processes in connection with individual prevention choices during epidemics, see e.g.,~\cite{franceschi2023PLOSONE,zanella2023BMB}.

In our case, instead, we do not consider a compartmental model but a genuinely kinetic model of a multi-agent society, whose microscopic state is the \textit{awareness} of the individuals, namely their ability to detect (possibly partly) fake news. The awareness may change in consequence of the interplay of the individuals with various pieces of information, according to a learning process which leads individuals to either increase or decrease their awareness about fake news. For its part, news acts as a background, that we model in detail by means of the statistical distribution of the news \textit{reliability}. Indeed, news may not only be completely true or false, as malevolent spreaders often combine true and false information to enhance the perceived credibility, thereby making the contents more compelling and likely to be accepted by the audience. Moreover, we imagine that the individual learning is affected also by the social context through the mean awareness of the population. Interestingly, although we do not assume a compartmentalisation of the society based on the inclination of the individuals to spread fake news, we recover this characteristic from our model as an emergent property of the system. In addition to this, taking inspiration from~\cite{toscani2018PRE}, we also analyse how individuals with benevolent intentions, when exposed to such contents, may, based on their awareness, contribute or not to the spread of possible misinformation on social media.

In more detail, the paper is organised as follows: Section~\ref{sect:awareness} introduces the kinetic model describing the awareness formation about fake news as a result of the learning process fostered by the interplay among individuals and news. The model is analysed qualitatively, whereby existence, uniqueness, and trend to equilibrium of its solutions are established. In addition to this, the unique attractive equilibrium profile of awareness, showing compartmentalisation as an emergent property, is found explicitly and illustrated through selected numerical tests. The section ends by proposing a conceptual way to infer the reliability of news and its probability distribution, which play a major role in the model definition, from real data. Section~\ref{sect:popularity} proposes a kinetic model addressing the evolution of the popularity of contents on social media based on the distribution of the awareness about fake news in the society. Also in this case, the model is analysed qualitatively and its asymptotic trends are investigated. A characterisation of the unique attractive equilibrium profile of popularity, conditioned to the reliability of the considered content, is provided in terms of its Fourier transform. Moreover, some features of this profile, such as its lower order statistical moments and the slimness or fatness of its tail, are obtained explicitly and shown numerically in connection with the connectivity distribution of the users of the social media. Finally, Section~\ref{sect:conclusions} draws some conclusions and briefly sketches possible research perspectives.

\section{Awareness evolution by learning}
\label{sect:awareness}
We consider a large population of individuals characterised by their ability to discern fake news, which in this paper we call \textit{awareness} and denote by a scalar variable $x\in [0,\,1]$. Specifically, $x=0$ stands for a null awareness, viz. a complete inability to detect fake news; whereas $x=1$ stands for a full awareness, viz. a (mostly ideal) full ability to unmask fake news. Parallelly, we describe the \textit{reliability} of news by another scalar variable $y\in [0,\,1]$, where $y=0$ stands for completely false information and $y=1$ for fully reliable information.

In the sequel, it will sometimes be useful to refer to the awareness and to the reliability of news as two random variables, denoted $X$ and $Y$, respectively, of which the $x$'s and the $y$'s are the realisations.

Let $f=f(x,t):[0,\,1]\times [0,\,+\infty)\to\R_+$ be the distribution of the awareness of the individuals at time $t$ and $g\in L^1(0,\,1)$, $g\geq 0$ a.e., the probability density function of the reliability of news, which we understand as prescribed and constant in time. In the sequel, it will be customary to consider both $f(\cdot,t)$ and $g$ defined on the whole $\R$ but with
\begin{equation}
	\supp{f(\cdot,t)}\subseteq [0,\,1] \quad \forall\,t\geq 0, \qquad \supp{g}\subseteq [0,\,1].
	\label{eq:supp}
\end{equation}
We further assume the following normalisation condition:
$$ \int_0^1f(x,t)\,dx=1 \quad \forall\,t\geq 0, $$
so that we may regard also $f(\cdot,t)$ as a probability density function. Furthermore, we denote by
$$ m_X(t):=\int_0^1xf(x,t)\,dx $$
the mean awareness of the population at time $t$.

To describe the evolution of the individual awareness against fake news in consequence of simple learning dynamics, we imagine that each interaction between an individual with level of awareness $x$ and a piece of information with reliability $y$ produces the post-interaction awareness
\begin{equation}
	x'=x+\alpha\lambda(x,y,t),
	\label{eq:x'}
\end{equation}
where $\alpha\in (0,\,1]$ is a given parameter and $\lambda:[0,\,1]^2\times [0,\,+\infty)\to [0,\,1]$ is the following learning function:
\begin{equation}
	\lambda(x,y,t)=
		\begin{cases}
			-x & \text{if } m_X(t)\leq y \\
			1-x & \text{if } m_X(t)>y.
		\end{cases}
	\label{eq:lambda}
\end{equation}
The meaning of the interaction rule~\eqref{eq:x'}-\eqref{eq:lambda} is clear:
\begin{itemize}
\item if the social context features a collective awareness $m_X$ lower than the reliability $y$ of news then $x'=(1-\alpha)x$. Individuals are not stimulated to increase their awareness against fake news and may instead lose part of it as a consequence of the addiction to fake news produced by the social context;
\item on the contrary, if the social context features a collective awareness $m_X$ higher than the reliability $y$ of news then $x'=x+\alpha(1-x)$. In this case, individuals are induced to increase their personal awareness against fake news by the social context, because the latter is, on the whole, sufficiently aware of them.
\end{itemize}
We remark that with~\eqref{eq:lambda} we compare the reliability of news with the mean awareness $m_X$ of the population and not e.g., with the individual awareness $x$. This allows us to take into account in a simple way the influence of the social context in the learning process.

The following consistency check is in order:
\begin{proposition} \label{prop:x'}
The interaction rule~\eqref{eq:x'}-\eqref{eq:lambda} with $\alpha\in (0,\,1]$ is physically admissible, namely $x'\in [0,\,1]$ for all $x,\,y\in [0,\,1]$ and all $t\geq 0$.
\end{proposition}
\begin{proof}
Writing $x'=(1-\alpha)x+\alpha\chi(m_X(t)>y)$, where $\chi(\cdot)$ is the characteristic function of the event in parenthesis, we see that $(1-\alpha)x\leq x'\leq (1-\alpha)x+\alpha$. Since $(1-\alpha)x\geq 0$ and $(1-\alpha)x+\alpha\leq 1-\alpha+\alpha=1$, the thesis follows.
\end{proof}

Proposition~\ref{prop:x'} ensures that the property~\eqref{eq:supp} of $\supp{f(\cdot,t)}$ is met for all $t>0$ provided it is at $t=0$.

\subsection{Kinetic description and trend to equilibrium}
\label{sect:kinetic_description.x}
By appealing to standard techniques, see e.g.,~\cite{pareschi2013BOOK} for details, it is possible to describe the time evolution of the probability density function $f$ subject to the interaction rule~\eqref{eq:x'}-\eqref{eq:lambda} by means of a Boltzmann-type kinetic equation. Since rule~\eqref{eq:x'}-\eqref{eq:lambda} is, in general, neither smooth nor invertible, it is convenient to refer to the weak form of the equation, which writes
\begin{equation}
	\frac{d}{dt}\int_0^1\varphi(x)f(x,t)\,dx=\int_0^1\int_0^1(\varphi(x')-\varphi(x))f(x,t)g(y)\,dx\,dy,
	\label{eq:Boltzmann.x}
\end{equation}
where $\varphi:[0,\,1]\to\C$ is an arbitrary observable quantity (test function) and $x'$ is given by~\eqref{eq:x'}. Owing to Proposition~\ref{prop:x'}, the term $\varphi(x')$ is well defined for every observable quantity $\varphi$.

Let
$$ G(y):=\int_{-\infty}^yg(u)\,du $$
be the cumulative distribution function of the reliability of news. Clearly, $G$ is non-decreasing; since $g\in L^1(0,\,1)$, it is also continuous. If, furthermore, $g\in L^1(0,\,1)\cap L^\infty(0,\,1)$ then $G$ is in particular Lipschitz continuous. In the sequel, we will invariably make this assumption and denote by $\Lip(G)>0$ the Lipschitz constant of $G$.

Choosing $\varphi(x)=x$ in~\eqref{eq:Boltzmann.x}, we obtain the following evolution equation for the mean awareness:
\begin{align}
	\begin{aligned}[b]
		\dot{m}_X &= \alpha\int_0^1\left(\int_0^1\lambda(x,y,t)g(y)\,dy\right)f(x,t)\,dx \\
		&= \alpha\int_0^1\left(-\int_{m_X}^1xg(y)\,dy+\int_0^{m_X}(1-x)g(y)\,dy\right)f(x,t)\,dx \\
		&= \alpha\left(G(m_X)-m_X\right),
	\end{aligned}
	\label{eq:mX}
\end{align}
where in the last passage we have used that $G(0)=0$, $G(1)=1$. 

If $G(y)=y$ in $[0,\,1]$, i.e. if $Y\sim\cU([0,\,1])$, then the mean awareness is conserved in time because the right-hand side of~\eqref{eq:mX} vanishes. We postpone this particular case to Subsection~\ref{sect:Y.unif} and, for the moment, we assume in general $G(y)\not\equiv y$ in $[0,\,1]$.

If $G$ does not coincide with the identity function in $[0,\,1]$ then, due to its monotonicity, there exists at most one $\bar{y}\in (0,\,1)$ such that $G(\bar{y})=\bar{y}$, hence there exists at most one equilibrium 
$$ m_X^\infty:=\bar{y}\in (0,\,1) $$
of~\eqref{eq:mX}. Notice that $0$, $1$ are instead always equilibria of~\eqref{eq:mX} for every $G$ because of the general properties of a cumulative distribution function recalled above. Depending on whether $0$, $m_X^\infty$, $1$ are stable and attractive equilibria, they represent the possible asymptotic values of the mean awareness emerging from the interactions between individuals and news. The stability and attractiveness of such equilibria may be ascertained by discussing the sign of the right-hand side of~\eqref{eq:mX}.

Specifically, assume that $\bar{y}\in (0,\,1)$ exists. Then, it is not difficult to conclude that:
\begin{enumerate}[label=(\roman*)]
\item if
\begin{equation}
	\begin{cases}
		G(y)>y & \text{for } y\in (0,\,\bar{y}) \\
		G(y)<y & \text{for } y\in (\bar{y},\,1)
	\end{cases}
	\label{eq:GY.stability}
\end{equation}
then $m_X^\infty$ is a stable and attractive equilibrium of~\eqref{eq:mX}, hence
$$\lim_{t\to +\infty}m_X(t)=m_X^\infty, \qquad \forall\,m_X^0\neq 0,\,1, $$
where we denote $m_X^0:=m_X(0)$;
\item if instead
\begin{equation*}
	\begin{cases}
		G(y)<y & \text{for } y\in (0,\,\bar{y}) \\
		G(y)>y & \text{for } y\in (\bar{y},\,1)
	\end{cases}
\end{equation*}
then $m_X^\infty$ is an unstable equilibrium of~\eqref{eq:mX} and
\begin{equation*}
	\lim_{t\to +\infty}m_X(t)=
		\begin{cases}
			0 & \text{if } m_X^0<m_X^\infty \\
			1 & \text{if } m_X^0>m_X^\infty.
		\end{cases}
\end{equation*}
\end{enumerate}
On the contrary, if $\bar{y}$ does not exist then
\begin{equation*}
	\lim_{t\to +\infty}m_X(t)=
		\begin{cases}
			0 & \text{if } G(y)<y,\quad \forall\,y\in (0,\,1)\ \text{and } m_X^0\neq 1 \\
			1 & \text{if } G(y)>y,\quad \forall\,y\in (0,\,1)\ \text{and } m_X^0\neq 0.
		\end{cases}
\end{equation*}

Summarising, we have proved that the mean awareness reaches an asymptotic value, which, depending on the statistical distribution of fake news, may be either
\begin{enumerate*}[label=(\roman*)]
\item settle on an intermediate value $0<m_X^\infty<1$ dictated by the cumulative distribution function $G$ of the reliability of news; or
\item coincide with the extreme values $0$, $1$, which represent a stylised description of a corrupted or virtuous polarisation of the collective awareness against fake news.
\end{enumerate*}
If $Y\sim\cU([0,\,1])$, i.e., if the reliability of news is uniformly distributed, then the mean awareness is conserved in time, thus $m_X(t)=m_X^0$ for all $t>0$.

The next result provides an estimate of continuous dependence of the mean awareness, which will be useful in the sequel.
\begin{lemma}
\label{lemma:mX}
Let $m_{X,1}(t),\,m_{X,2}(t)$ be solutions to~\eqref{eq:mX} in the interval $(T,\,+\infty)$, issuing from initial values $m_{X,1}(T),\,m_{X,2}(T)\in [0,\,1]$, respectively, where $T\geq 0$ is arbitrary. Then
$$ \abs{m_{X,2}(t)-m_{X,1}(t)}\leq\abs{m_{X,2}(T)-m_{X,1}(T)}e^{\alpha\left(\Lip(G)-1\right)(t-T)}, \quad t>T. $$
\end{lemma}
\begin{proof}
From~\eqref{eq:mX}, taking the difference between the equations satisfied by $m_{X,1}$, $m_{X,2}$, we get
$$ \frac{d}{dt}(m_{X,2}-m_{X,1})+\alpha(m_{X,2}-m_{X,1})=\alpha\left(G(m_{X,2})-G(m_{X,1})\right), $$
whence, multiplying both sides by $e^{\alpha t}$ and integrating in time over the interval $[T,\,t]$, $t>T$,
\begin{align*}
	e^{\alpha t}\left(m_{X,2}(t)-m_{X,1}(t)\right) &= e^{\alpha T}\left(m_{X,2}(T)-m_{X,1}(T)\right) \\
	&\phantom{=} +\alpha\int_T^te^{\alpha\tau}\bigl(G(m_{X,2}(\tau))-G(m_{X,1}(\tau))\bigr)\,d\tau.
\end{align*}
The Lipschitz continuity of $G$ allows us to deduce
\begin{align*}
	\abs{e^{\alpha t}\left(m_{X,2}(t)-m_{X,1}(t)\right)} &\leq e^{\alpha T}\abs{m_{X,2}(T)-m_{X,1}(T)} \\
	&\phantom{=} +\alpha\Lip(G)\int_T^t\abs{e^{\alpha\tau}(m_{X,2}(\tau)-m_{X,1}(\tau))}\,d\tau
\end{align*}
and Gr\"{o}nwall's inequality implies
$$ \abs{e^{\alpha t}\left(m_{X,2}(t)-m_{X,1}(t)\right)}\leq\abs{m_{X,2}(T)-m_{X,1}(T)}e^{\alpha T}e^{\alpha\Lip(G)(t-T)}, $$
whence the thesis follows.
\end{proof}

The behaviour of $m_X$ is at the basis of the trend to equilibrium of the solutions to~\eqref{eq:Boltzmann.x}, hence of the identification of the asymptotic distributions\footnote{The so-called \textit{Maxwellians} in the jargon of the classical kinetic theory.} of the awareness against fake news. To approach this issue, we begin by investigating the trend of the energy of the awareness distribution:
$$ E_X(t):=\int_0^1x^2f(x,t)\,dx; $$
in particular, plugging $\varphi(x)=x^2$ into~\eqref{eq:Boltzmann.x} we discover
\begin{equation}
	\dot{E}_X=\alpha(2-\alpha)\left(\frac{2(1-\alpha)m_X+\alpha}{2-\alpha}G(m_X)-E_X\right).
	\label{eq:EX}
\end{equation}
Using Lemma~\ref{lemma:mX}, it is not difficult to see that when the equilibrium $m_X^\infty\in (0,\,1)$ of~\eqref{eq:mX} exists and is asymptotically stable, i.e., when $G$ satisfies~\eqref{eq:GY.stability}, the convergence of $m_X$ to $m_X^\infty$ is exponentially fast in time. Consequently, also $\frac{2(1-\alpha)m_X+\alpha}{2-\alpha}G(m_X)$ converges to $\frac{2(1-\alpha)m_X^\infty+\alpha}{2-\alpha}G(m_X^\infty)=\frac{2(1-\alpha)m_X^\infty+\alpha}{2-\alpha}m_X^\infty$ exponentially fast, because
$$ \abs*{\frac{2(1-\alpha)m_X+\alpha}{2-\alpha}G(m_X)-\frac{2(1-\alpha)m_X^\infty+\alpha}{2-\alpha}G(m_X^\infty)}\leq \left(\Lip(G)+\frac{2(1-\alpha)}{2-\alpha}\right)\abs{m_X-m_X^\infty}, $$
and in such a case it is known that the asymptotic behaviour of the solution to~\eqref{eq:EX} is
$$ \lim_{t\to +\infty}E_X(t)=\frac{2(1-\alpha)m_X^\infty+\alpha}{2-\alpha}m_X^\infty=:E_X^\infty. $$
Notice that $(m_X^\infty)^2<E_X^\infty\leq m_X^\infty$, therefore the internal energy, i.e., the variance of distribution $f$, is in general non-zero asymptotically. This implies that the solutions to~\eqref{eq:Boltzmann.x} possibly evolve towards non-trivial equilibrium distributions. The only trivial cases are when $m_X\to 0,\,1$, for then also $E_X\to 0,\,1$ and the asymptotic distribution is either $\delta_0$ or $\delta_1$, where $\delta_a$ denotes the Dirac distribution centred at $x=a\in\R$.

For a more detailed investigation of the trend to equilibrium we introduce now some tools from measure theory, which have proved to be particularly effective in the analysis of kinetic equations.

Let $\cP([0,\,1])$ be the set of probability measures $\mu$ supported in $[0,\,1]\subset\R$. We denote by
$$ \hat{\mu}(\xi):=\int_\R e^{-i\xi x}\,d\mu(x) $$
the Fourier transform of any such $\mu$, which is a bounded and continuous function on $\R$. Here and henceforth, $i$ denotes the imaginary unit. Since $\supp{\mu}\subseteq [0,\,1]$, we may equivalently write
$$ \hat{\mu}(\xi)=\int_{[0,\,1]}e^{-i\xi x}\,d\mu(x). $$
Given $\mu,\,\nu\in\cP([0,\,1])$, we define their Fourier distance $d_s$, $s\in (0,\,1]$, as
$$ d_s(\mu,\nu):=\sup_{\xi\in\R\setminus\{0\}}\frac{\abs{\hat{\mu}(\xi)-\hat{\nu}(\xi)}}{\abs{\xi}^s}. $$
For a thorough review of Fourier metrics and their properties in connection with kinetic equations we refer the interested reader to~\cite{carrillo2007RMUP}; see also~\cite{auricchio2020RLMA,spiga2004AML}. For our purposes, we mention that the Fourier distance $d_s$, $s>0$, between any two probability measures is finite provided the statistical moments of the two measures coincide up to the order $[s]$, i.e. the integer part of $s$, if $s\in\R\setminus\N$ and up to the order $s-1$ if $s\in\N$. In our case, we confine ourselves to $s\leq 1$ because, in general, the statistical moments of the solutions to~\eqref{eq:Boltzmann.x} differ from the order $1$ onwards. The reason is that, except when $Y\sim\cU([0,\,1])$, the mean awareness $m_X$ is not conserved by the learning dynamics~\eqref{eq:x'}-\eqref{eq:lambda}.

In correspondence of any prescibed initial datum $f^0\in\cP([0,\,1])$,~\eqref{eq:Boltzmann.x} admits a unique global solution $f(t)=f(\cdot,t)\in\cP([0,\,1])$, $t>0$. We prove precisely this statement in Theorem~\ref{theo:exist_uniq} of Appendix~\ref{sect:proof}. Here, building on this result, we move to the investigation of the large time trend of the solutions:

\begin{theorem}
\label{theo:ds}
Assume that $G\in C^1(0,\,1)$ satisfies~\eqref{eq:GY.stability}, hence that there exists a unique asymptotically stable equilibrium $m_X^\infty$ of~\eqref{eq:mX} in $(0,\,1)$. If $\alpha\in (0,\,1]$ is such that
\begin{equation}
	\frac{(1-\alpha)^s}{\alpha}<1-g(m_X^\infty)
	\label{ass:alpha}
\end{equation}
for some $s\in (0,\,1]$ then any two solutions\footnote{Issuing e.g., from two different initial conditions.} $f_1(t),\,f_2(t)\in\cP([0,\,1])$, $t>0$, to~\eqref{eq:Boltzmann.x} with initial means $m_{X,1}^0,\,m_{X,2}^0\in (0,\,1)$ are such that
$$ \lim_{t\to +\infty}d_s(f_1(t),f_2(t))=0. $$
\end{theorem}

\begin{remark}
Owing to~\eqref{eq:GY.stability}, $g(m_X^\infty)=G'(m_X^\infty)\leq 1$. In particular, the $\alpha$'s fulfilling~\eqref{ass:alpha} form a non-empty subset of $(0,\,1]$ whenever $g(m_X^\infty)<1$, for then $1-g(m_X^\infty)>0$. Notice that the less tight condition on $\alpha$ is obtained with $s=1$.
\end{remark}

\begin{proof}[Proof of Theorem~\ref{theo:ds}]
Taking $\varphi(x)=e^{-i\xi x}$ in~\eqref{eq:Boltzmann.x} and invoking~\eqref{eq:x'},~\eqref{eq:lambda} yields, after some algebraic manipulations,
\begin{align*}
	\partial_t\hat{f}(\xi,t) &= \int_0^1\left(\int_{m_X}^1e^{-i\xi(1-\alpha)x}g(y)\,dy+\int_0^{m_X}e^{-i\xi(\alpha+(1-\alpha)x)}g(y)\,dy\right)f(x,t)\,dx-\hat{f}(\xi,t) \\
	&= \left(1-G(m_X)\right)\hat{f}((1-\alpha)\xi,t)+e^{-i\alpha\xi}G(m_X)\hat{f}((1-\alpha)\xi,t)-\hat{f}(\xi,t),
\end{align*}
where we have used $G(0)=0$, $G(1)=1$. Upon setting
\begin{equation}
	H(m_X,\xi):=1+\left(e^{-i\alpha\xi}-1\right)G(m_X),
	\label{eq:H}
\end{equation}
we rewrite this equation compactly as
$$ \partial_t\hat{f}=H(m_X,\xi)\hat{f}((1-\alpha)\xi,t)-\hat{f}. $$
This gives the time evolution of the Fourier transform of a generic solution to~\eqref{eq:Boltzmann.x}. We now apply it to $f_1$, $f_2$ and take the difference to obtain
$$ \partial_t(\hat{f}_1-\hat{f}_2)=H(m_{X,1},\xi)\hat{f}_1((1-\alpha)\xi,t)-H(m_{X,2},\xi)\hat{f}_2((1-\alpha)\xi,t)-(\hat{f}_1-\hat{f}_2). $$
Dividing both sides by $\abs{\xi}^s$ we further get
$$ \partial_t\frac{\hat{f}_1-\hat{f}_2}{\abs{\xi}^s}=\frac{H(m_{X,1},\xi)\hat{f}_1((1-\alpha)\xi,t)-H(m_{X,2},\xi)\hat{f}_2((1-\alpha)\xi,t)}{\abs{\xi}^s}-\frac{\hat{f}_1-\hat{f}_2}{\abs{\xi}^s}, $$
i.e., setting $h(\xi,t):=\frac{\hat{f}_1(\xi,t)-\hat{f}_2(\xi,t)}{\abs{\xi}^s}$ for brevity,
$$ \partial_th+h=\frac{H(m_{X,1},\xi)\hat{f}_1((1-\alpha)\xi,t)-H(m_{X,2},\xi)\hat{f}_2((1-\alpha)\xi,t)}{\abs{\xi}^s}. $$
Multiplying both sides by $e^t$ and integrating in time over an interval of the form $[T,\,t]$, where $T>0$ is fixed and $t>T$, yields
$$ e^th(\xi,t)=e^Th(\xi,T)+\int_T^te^\tau\frac{H(m_{X,1},\xi)\hat{f}_1((1-\alpha)\xi,\tau)-H(m_{X,2},\xi)\hat{f}_2((1-\alpha)\xi,\tau)}{\abs{\xi}^s}\,d\tau, $$
and further
\begin{equation}
	\abs{e^th(\xi,t)}\leq\abs{e^Th(\xi,T)}+\int_T^te^\tau\frac{\abs{H(m_{X,1},\xi)\hat{f}_1((1-\alpha)\xi,\tau)-H(m_{X,2},\xi)\hat{f}_2((1-\alpha)\xi,\tau)}}{\abs{\xi}^s}\,d\tau.
	\label{eq:|et.h|}
\end{equation}

Now, observe that
\begin{align*}
	&\frac{\abs{H(m_{X,1},\xi)\hat{f}_1((1-\alpha)\xi,t)-H(m_{X,2},\xi)\hat{f}_2((1-\alpha)\xi,t)}}{\abs{\xi}^s}\leq \\
	&\hspace{4cm} \leq\abs{H(m_{X,1},\xi)}\cdot\frac{\abs{\hat{f}_1((1-\alpha)\xi,t)-\hat{f}_2((1-\alpha)\xi,t)}}{\abs{\xi}^s} \\
	&\hspace{4cm}\phantom{\leq} +\frac{\abs{H(m_{X,1},\xi)-H(m_{X,2},\xi)}}{\abs{\xi}^s}\cdot\abs{\hat{f}_2((1-\alpha)\xi,t)}
\end{align*}
and that $\abs{H(m_X,\xi)}\leq 1-G(m_X)+G(m_X)=1$, while
$$ \abs{\hat{f}_2((1-\alpha)\xi,t)}=\abs*{\int_0^1f_2(x,t)e^{-i(1-\alpha)\xi x}\,dx}\leq\int_0^1f_2(x,t)\,dx=1, $$
hence
\begin{multline*}
	\frac{\abs{H(m_{X,1},\xi)\hat{f}_1((1-\alpha)\xi,t)-H(m_{X,2},\xi)\hat{f}_2((1-\alpha)\xi,t)}}{\abs{\xi}^s}\leq \\
	\leq \frac{\abs{\hat{f}_1((1-\alpha)\xi,t)-\hat{f}_2((1-\alpha)\xi,t)}}{\abs{\xi}^s}+\frac{\abs{H(m_{X,1},\xi)-H(m_{X,2},\xi)}}{\abs{\xi}^s}.
\end{multline*}

Let us examine more closely the second term at the right-hand side. Recalling the definition~\eqref{eq:H} of $H$ and using the Lipschitz continuity of $G$ in $[0,\,1]$, we have
\begin{align*}
	\frac{\abs{H(m_{X,1},\xi)-H(m_{X,2},\xi)}}{\abs{\xi}^s} &= \abs{G(m_{X,1})-G(m_{X,2})}\cdot\frac{\abs{e^{-i\alpha\xi}-1}}{\abs{\xi}^s} \\
	&\leq \Lip(G)\abs{m_{X,1}-m_{X,2}}\cdot\frac{\abs{e^{-i\alpha\xi}-1}}{\abs{\xi}^s}.
\end{align*}
Moreover,
$$ \frac{\abs{e^{-i\alpha\xi}-1}}{\abs{\xi}^s}=\sqrt{2\frac{1-\cos{(\alpha\xi)}}{\abs{\xi}^{2s}}}\leq 2^{1-s}\alpha^s $$
for every $\xi\in\R$, therefore we conclude
\begin{align*}
	\frac{\abs{H(m_{X,1},\xi)\hat{f}_1((1-\alpha)\xi,t)-H(m_{X,2},\xi)\hat{f}_2((1-\alpha)\xi,t)}}{\abs{\xi}^s} &\leq \frac{\abs{\hat{f}_1((1-\alpha)\xi,t)-\hat{f}_2((1-\alpha)\xi,t)}}{\abs{\xi}^s} \\
	&\phantom{\leq} +2^{1-s}\alpha^s\Lip(G)\abs{m_{X,1}-m_{X,2}}.
\end{align*}

Back to~\eqref{eq:|et.h|}, this implies
\begin{align*}
	\abs{e^th(\xi,t)} &\leq \abs{e^Th(\xi,T)}+2^{1-s}\alpha^s\Lip(G)\int_T^te^\tau\abs{m_{X,1}(\tau)-m_{X,2}(\tau)}\,d\tau \\
	&\phantom{\leq} +\int_T^te^\tau\frac{\abs{\hat{f}_1((1-\alpha)\xi,\tau)-\hat{f}_2((1-\alpha)\xi,\tau)}}{\abs{\xi}^s}\,d\tau.
\end{align*}
Notice that
$$ d_s(f_1(t),f_2(t))=\sup_{\xi\in\R\setminus\{0\}}\abs{h(\xi,t)}=:\norminf{h(t)} $$
and that
\begin{align*}
	\sup_{\xi\in\R\setminus\{0\}}\frac{\abs{\hat{f}_1((1-\alpha)\xi,t)-\hat{f}_2((1-\alpha)\xi,t)}}{\abs{\xi}^s} &=
		(1-\alpha)^s\sup_{\xi\in\R\setminus\{0\}}\frac{\abs{\hat{f}_1((1-\alpha)\xi,t)-\hat{f}_2((1-\alpha)\xi,t)}}{\abs{(1-\alpha)\xi}^s} \\
	&= (1-\alpha)^s\sup_{\eta\in\R\setminus\{0\}}\frac{\abs{\hat{f}_1(\eta,t)-\hat{f}_2(\eta,t)}}{\abs{\eta}^s} \\
	&= (1-\alpha)^s\norminf{h(t)},
\end{align*}
whence
\begin{align*}
	\abs{e^th(\xi,t)} &\leq \abs{e^Th(\xi,T)}+2^{1-s}\alpha^s\Lip(G)\int_T^te^\tau\abs{m_{X,1}(\tau)-m_{X,2}(\tau)}\,d\tau \\
	&\phantom{\leq} +(1-\alpha)^s\int_T^t\norminf{e^\tau h(\tau)}\,d\tau.
\end{align*}
Taking the supremum of both sides over $\xi\in\R\setminus\{0\}$ we arrive at
\begin{align*}
	\norminf{e^th(t)} &\leq \norminf{e^Th(T)}+2^{1-s}\alpha^s\Lip(G)\int_T^te^\tau\abs{m_{X,1}(\tau)-m_{X,2}(\tau)}\,d\tau \\
	&\phantom{\leq} +(1-\alpha)^s\int_T^t\norminf{e^\tau h(\tau)}\,d\tau,
\end{align*}
which, owing to Lemma~\ref{lemma:mX}, we further develop as
\begin{align*}
	\norminf{e^th(t)} &\leq \norminf{e^Th(T)}+C\int_T^te^{\left[\alpha\left(\Lip(G)-1\right)+1\right]\tau}\,d\tau+(1-\alpha)^s\int_T^t\norminf{e^\tau h(\tau)}\,d\tau \\
	&= \norminf{e^Th(T)}+C'\left(e^{\left[\alpha\left(\Lip(G)-1\right)+1\right]t}-e^{\left[\alpha\left(\Lip(G)-1\right)+1\right]T}\right) \\
	&\phantom{=} +(1-\alpha)^s\int_T^t\norminf{e^\tau h(\tau)}\,d\tau,
\end{align*}
where $C,\,C'>0$ are constants. Invoking Gr\"{o}nwall's inequality, we discover then\footnote{For $a,\,b\geq 0$, we write $a\lesssim b$ to mean that there exists a constant $K>0$, whose specific value is unimportant, such that $a\leq Kb$.}
\begin{align}
	\begin{aligned}[b]
		\norminf{h(t)} &\lesssim \left(\norminf{e^Th(T)}+e^{\left[\alpha\left(\Lip(G)-1\right)+1\right]t}-e^{\left[\alpha\left(\Lip(G)-1\right)+1\right]T}\right)e^{-[1-(1-\alpha)^s]t} \\
		&= \left(\norminf{e^Th(T)}-e^{\left[\alpha\left(\Lip(G)-1\right)+1\right]T}\right)e^{-[1-(1-\alpha)^s]t}+e^{\left[\alpha\left(\Lip(G)-1\right)+(1-\alpha)^s\right]t}.
	\end{aligned}
	\label{eq:norm.h}
\end{align}

Let us examine the exponent $\alpha\left(\Lip(G)-1\right)+(1-\alpha)^s$ of the second term on the right-hand side. Since, by assumption, the equilibrium $m_X^\infty\in (0,\,1)$ of~\eqref{eq:mX} exists and is attractive, and moreover $m_{X,1}^0,\,m_{X,2}^0\neq 0,\,1$, for every $\eta>0$ there exists $T>0$ such that $m_{X,1}(t),\,m_{X,2}(t)\in (m_X^\infty-\eta,\,m_X^\infty+\eta)$ for all $t>T$. Thus, it is sufficient to estimate $\Lip(G)$ in a neighbourhood of the form $(\bar{y}-\eta,\,\bar{y}+\eta)$, $\bar{y}=m_X^\infty$. Recalling that $G'_Y(y)=g(y)$, we have
$$ \Lip(G)=\sup_{y\in (\bar{y}-\eta,\,\bar{y}+\eta)}g(y). $$
In particular, owing to the continuity of $g$ in $[0,\,1]$ (because $G\in C^1(0,\,1)$ by assumption), we can choose $\eta$ so small that $\Lip(G)=g(\bar{y})+\epsilon$ for some $\epsilon>0$. Hence
$$ \alpha\left(\Lip(G)-1\right)+(1-\alpha)^s=\alpha\left(g(\bar{y})+\epsilon-1\right)+(1-\alpha)^s $$
and if we fix $\epsilon<1-g(\bar{y})-(1-\alpha)^s/\alpha$ we obtain $\alpha\left(\Lip(G)-1\right)+(1-\alpha)^s<0$. Notice that such an $\epsilon$ exists because $1-g(\bar{y})-(1-\alpha)^s/\alpha>0$ by~\eqref{ass:alpha}.

Finally, passing to the limit $t\to +\infty$ in~\eqref{eq:norm.h} we get $\norminf{h(t)}\to 0$ and we are done.
\end{proof}

\subsubsection{Identification of the Maxwellian}
\label{sect:Maxwellian.x}
As seen in the proof of Theorem~\ref{theo:ds}, the Fourier-transformed version of~\eqref{eq:Boltzmann.x} reads
$$ \partial_t\hat{f}(\xi,t)=H(m_X,\xi)\hat{f}((1-\alpha)\xi,t)-\hat{f}(\xi,t), $$
where $H$ is defined by~\eqref{eq:H}. If the equilibrium $m_X^\infty\in (0,\,1)$ of~\eqref{eq:mX} exists and is asymptotically stable, cf.~\eqref{eq:GY.stability}, we may construct a stationary solution $f^\infty=f^\infty(x)$ to~\eqref{eq:Boltzmann.x} with mean $m_X^\infty$ by setting
$$ H(m_X^\infty,\xi)\hat{f}^\infty((1-\alpha)\xi)-\hat{f}^\infty(\xi)=0, $$
whence
\begin{equation}
	\hat{f}^\infty(\xi)=H(m_X^\infty,\xi)\hat{f}^\infty((1-\alpha)\xi).
	\label{eq:hatf.recur}
\end{equation}
From this, arguing recursively, we get:
\begin{proposition}
\label{prop:Maxwellian.x}
The Maxwellian $f^\infty\in\cP([0,\,1])$ is univocally characterised by its Fourier transform as
\begin{equation}
	\hat{f}^\infty(\xi)=\prod_{k=0}^{\infty}H(m_X^\infty,(1-\alpha)^k\xi),
	\label{eq:hatf}
\end{equation}
where $H$ is given by~\eqref{eq:H}.
\end{proposition}
\begin{proof}
First, we check that every $\hat{f}^\infty$ of the form $\hat{f}^\infty(\xi)=C\prod_{k=0}^{\infty}H(m_X^\infty,(1-\alpha)^k\xi)$, where $C>0$ is a constant, satisfies~\eqref{eq:hatf.recur}. We have:
\begin{align*}
	H(m_X^\infty,\xi)\hat{f}^\infty((1-\alpha)\xi) &= CH(m_X^\infty,\xi)\prod_{k=0}^{\infty}H(m_X^\infty,(1-\alpha)^{k+1}\xi) \\
	&= CH(m_X^\infty,\xi)\prod_{k=1}^{\infty}H(m_X^\infty,(1-\alpha)^k\xi) \\
	&= C\prod_{k=0}^{\infty}H(m_X^\infty,(1-\alpha)^k\xi)=\hat{f}^\infty(\xi).
\end{align*}
Moreover, since $\hat{f}^\infty(0)=\int_0^1f^\infty(x)\,dx=1$ and $H(m_X^\infty,0)=1$ (cf.~\eqref{eq:H}), it follows $C=1$.

Second, we assume that $\hat{f}^\infty$ satisfies~\eqref{eq:hatf.recur} and we show that it has necessarily the form~\eqref{eq:hatf}. For this, we define:
$$ F(\xi):=\frac{\hat{f}^\infty(\xi)}{\prod\limits_{k=0}^{\infty}H(m_X^\infty,(1-\alpha)^k\xi)} $$
and, invoking~\eqref{eq:hatf.recur}, we observe that
$$ F(\xi)=\frac{H(m_X^\infty,\xi)\hat{f}^\infty((1-\alpha)\xi)}{\prod\limits_{k=0}^{\infty}H(m_X^\infty,(1-\alpha)^k\xi)}=
	\frac{\hat{f}^\infty((1-\alpha)\xi)}{\prod\limits_{k=1}^{\infty}H(m_X^\infty,(1-\alpha)^k\xi)}=F((1-\alpha)\xi), $$
whence $F(\xi)=F((1-\alpha)^n\xi)$ for all $n\in\N$ and all $\xi\in\R$. Passing to the limit $n\to\infty$ we discover, by continuity of $F$ and using the fact that $0\leq 1-\alpha<1$,
$$ F(\xi)=\lim_{n\to\infty}F((1-\alpha)^n\xi)=F(0)=1, \quad \forall\,\xi\in\R. $$
Thus $F$ is constant and the thesis follows.
\end{proof}

When the equilibrium $m_X^\infty\in (0,\,1)$ of~\eqref{eq:mX} exists and is asymptotically stable, Theorem~\ref{theo:ds} ensures that the Maxwellian constructed in Proposition~\ref{prop:Maxwellian.x} is the only stationary distribution towards which solutions to~\eqref{eq:Boltzmann.x} with initial mean $m_X^0\neq 0,\,1$ converge in time, at least in a suitable range of values of the parameter $\alpha$, cf.~\eqref{ass:alpha}. Indeed:
\begin{corollary}
\label{cor:asymptotics.x}
Under the same assumptions as in Theorem~\ref{theo:ds},
$$ \lim_{t\to +\infty}d_s(f(t),f^\infty)=0, $$
where $f(t)\in\cP([0,\,1])$ is any solution to~\eqref{eq:Boltzmann.x} issuing from an initial condition with mean $m_X^0\neq 0,\,1$ and $f^\infty$ is the stationary solution built in Proposition~\ref{prop:Maxwellian.x}.
\end{corollary}
\begin{proof}
The result follows from Theorem~\ref{theo:ds} by taking $f_2=f^\infty$, which is indeed a (constant-in-time) solution to~\eqref{eq:Boltzmann.x}.
\end{proof}

For the sake of completeness, we observe that if $m_X^0=0$ then $f^\infty=\delta_0$, for in that case the mean awareness remains constant and equal to $0$ at all times $t>0$. For an analogous reason, if $m_X^0=1$ then $f^\infty=\delta_1$. If instead $m_X^\infty$ does not exist or is unstable (and, in the latter case, $m_X^0\neq m_X^\infty$) then $f^\infty$ coincides with the Dirac distribution centred in either $0$ or $1$, depending on which of these two values is the asymptotically stable equilibrium of~\eqref{eq:mX}, or, if they are both, on whether $m_X^0\lessgtr m_X^\infty$. Finally, if $m_X^\infty$ is unstable and $m_X^0=m_X^\infty$ then by carefully inspecting the proof of Theorem~\ref{theo:ds} and arguing like in Corollary~\ref{cor:asymptotics.x} we see that the Maxwellian is again the one given by Proposition~\ref{prop:Maxwellian.x}, for in that case we can assume $m_{X,1}(t)=m_{X,2}(t)=m_X^\infty$ for all $t>0$.

In general, however, formula~\eqref{eq:hatf} provided by Proposition~\ref{prop:Maxwellian.x} does not allow one to write explicitly the Maxwellian in the space of the microscopic state $x$. This is possible only in very special cases, such as e.g., $\alpha=1$, for which the infinite product in~\eqref{eq:hatf} reduces to a finite one. Specifically, using the convention $0^0=1$ and considering that $H(m_X^\infty,(1-\alpha)^k\xi)=H(m_X^\infty,0)=1$ by~\eqref{eq:H} for all $k\geq 1$, we obtain
$$ \hat{f}^\infty(\xi)=H(m_X^\infty,\xi)=1+\left(e^{-i\xi}-1\right)G(m_X^\infty), $$
whence the inverse Fourier transform produces
\begin{equation}
	f^\infty=\left(1-G(m_X^\infty)\right)\delta_0+G(m_X^\infty)\delta_1.
	\label{eq:f^inf.x}
\end{equation}
This is a convex combination of $\delta_0$, $\delta_1$ representing a bipartite splitting of the population in a fraction $1-G(m_X^\infty)$ of individuals totally unable to detect fake news and the complementary fraction $G(m_X^\infty)$ of individuals extremely able to unmask them.

For $0<\alpha<1$, if we let
\begin{align*}
	\hat{f}^\infty_k(\xi) &:= H(m_X^\infty,(1-\alpha)^k\xi) \\
	&= 1+\left(e^{-i\alpha(1-\alpha)^k\xi}-1\right)G(m_X^\infty), \quad k\geq 0,
\end{align*}
we may observe that $\hat{f}^\infty_k$ is the Fourier transform of
$$ f^\infty_k=\left(1-G(m_X^\infty)\right)\delta_0+G(m_X^\infty)\delta_{\alpha(1-\alpha)^k} $$
and that, in view of~\eqref{eq:hatf}, we may construct $f^\infty$ as the infinite convolution of the $f^\infty_k$'s. Concerning this, notice that the $s$-Fourier distance, $s\in (0,\,1]$, between $f^\infty_k$ and $\delta_0$ diminishes for increasing $k$:
\begin{align}
    \begin{aligned}[b]
	d_s(f^\infty_k,\delta_0)=\sup_{\xi\in\R\setminus\{0\}}\frac{\abs{\hat{f}^\infty_k(\xi)-1}}{\abs{\xi}^s}
		&=G(m_X^\infty)\sup_{\xi\in\R\setminus\{0\}}\frac{\abs{e^{-i\alpha(1-\alpha)^k\xi}-1}}{\abs{\xi}^s} \\
	&= G(m_X^\infty)\sup_{\xi\in\R\setminus\{0\}}\sqrt{2\frac{1-\cos{(\alpha(1-\alpha)^k\xi)}}{\abs{\xi}^{2s}}} \\
	&= G(m_X^\infty)\alpha^s{(1-\alpha)}^{sk}\sup_{\eta\in\R\setminus\{0\}}\sqrt{2\frac{1-\cos{\eta}}{\abs{\eta}^{2s}}} \\
	&\leq G(m_X^\infty)2^{1-s}\alpha^s{(1-\alpha)}^{sk},
    \end{aligned}
    \label{eq:dist.fk-d0}
\end{align}
In particular, $d_s(f^\infty_k,\delta_0)\to 0$ exponentially fast for every $s\in (0,\,1]$ when $k\to\infty$. The closer $\alpha$ to $1$ the faster the convergence. Since the Dirac distribution is the identity element of the convolution, we expect a truncation of the infinite product in~\eqref{eq:hatf} to possibly provide a reliable approximation of $\hat{f}^\infty$, hence also of $f^\infty$ by inverse Fourier transform, in the sense of the metric $d_s$.

For example, considering the first three terms corresponding to $k=0,\,1,\,2$ and recalling that $\delta_a\ast\delta_b=\delta_{a+b}$ for $a,\,b\in\R$, where $\ast$ denotes convolution, we get:
\begin{align*}
	f^\infty &\approx f_0^\infty\ast f_1^\infty\ast f_2^\infty \\
	&= \left(1-G(m_X^\infty)\right)^3\delta_0 \\
	&\phantom{=} +G(m_X^\infty)\left(1-G(m_X^\infty)\right)^2\left[\delta_\alpha+\delta_{\alpha(1-\alpha)}+\delta_{\alpha(1-\alpha)^2}\right] \\
	&\phantom{=} +G^2(m_X^\infty)\left(1-G(m_X^\infty)\right)\left[\delta_{\alpha+\alpha(1-\alpha)}+\delta_{\alpha+\alpha(1-\alpha)^2}+\delta_{\alpha(1-\alpha)+\alpha(1-\alpha)^2}\right] \\
	&\phantom{=} +G^3(m_X^\infty)\delta_{\alpha+\alpha(1-\alpha)+\alpha(1-\alpha)^2},
\end{align*}
which shows that $f^\infty$ may be approximated by a weighted sum of Dirac distributions centred in the points
\begin{equation}
	\begin{aligned}
		x^\infty_1 &= 0 \\
		x^\infty_2 &= \alpha \\
		x^\infty_3 &= \alpha(1-\alpha) \\
		x^\infty_4 &= \alpha(1-\alpha)^2 \\
		x^\infty_5 &= \alpha+\alpha(1-\alpha) \\
		x^\infty_6 &= \alpha+\alpha(1-\alpha)^2 \\
		x^\infty_7 &= \alpha(1-\alpha)+\alpha(1-\alpha)^2 \\
		x^\infty_8 &= \alpha+\alpha(1-\alpha)+\alpha(1-\alpha)^2
	\end{aligned}
	\begin{aligned}
		&\vphantom{
			\begin{aligned}
				x^\infty_1 &= 0
			\end{aligned}
		} \quad\quad\ \text{with weight } \left(1-G(m_X^\infty)\right)^3 \\
		&\left.
		\vphantom{
			\begin{aligned}
				x^\infty_2 &= \alpha \\
				x^\infty_3 &= \alpha(1-\alpha) \\
				x^\infty_4 &= \alpha(1-\alpha)^2
			\end{aligned}
		}\right\} \quad \text{with weight } G(m_X^\infty)\left(1-G(m_X^\infty)\right)^2 \\
		&\left.
		\vphantom{
			\begin{aligned}
				x^\infty_5 &= \alpha+\alpha(1-\alpha) \\
				x^\infty_6 &= \alpha+\alpha(1-\alpha)^2 \\
				x^\infty_7 &= \alpha(1-\alpha)+\alpha(1-\alpha)^2
			\end{aligned}
		}\right\} \quad \text{with weight } G^2(m_X^\infty)\left(1-G(m_X^\infty)\right) \\
		&\vphantom{
			\begin{aligned}
				x^\infty_8 &= \alpha+\alpha(1-\alpha)+\alpha(1-\alpha)^2
			\end{aligned}
		} \quad\quad\ \text{with weight } G^3(m_X^\infty).
	\end{aligned}
	\label{eq:x.infty}
\end{equation}
Clearly, the accuracy of such an approximation depends on the value of $\alpha$, which, as shown by the computations above, determines the speed of convergence of $d_s(f^\infty_k,\delta_0)$ to $0$ for $k$ large. In general, however, we infer that the asymptotic profile of the awareness distribution resulting from the learning dynamics~\eqref{eq:x'}-\eqref{eq:lambda} features \textit{awareness clusters} distributed in $[0,\,1]$. In the special case $\alpha=1$, all clusters collapse in $\{0,\,1\}$.

\subsubsection{The case~\texorpdfstring{$\boldsymbol{Y\sim\cU([0,\,1])}$}{}}
\label{sect:Y.unif}
If the reliability of news is uniformly distributed in $[0,\,1]$, i.e. if $g(y)=\chi_{[0,\,1]}(y)$ and consequently
$$ G(y)=
	\begin{cases}
		0 & \text{if } y<0 \\
		y & \text{if } 0\leq y\leq 1 \\
		1 & \text{if } y>1,
	\end{cases} $$
then, as already observed,~\eqref{eq:mX} reduces to
$$ \dot{m}_X=0, $$
which indicates that the mean awareness is conserved by the learning dynamics~\eqref{eq:x'}-\eqref{eq:lambda}. In this case, every solution to~\eqref{eq:Boltzmann.x} preserves the initial mean, hence in particular $m_X^\infty=m_X^0$.

In view of Proposition~\ref{prop:Maxwellian.x}, we construct a Maxwellian $f^\infty$ with the prescribed mean $m_X^0\in [0,\,1]$ as:
\begin{equation}
	\hat{f}^\infty(\xi)=\prod_{k=0}^{\infty}H(m_X^0,(1-\alpha)^k\xi).
	\label{eq:hatf.mX0}
\end{equation}
Next, we strengthen the result of Corollary~\ref{cor:asymptotics.x} as follows (cf.~\cite{toscani2009EPL}):
\begin{theorem}
\label{theo:asymptotics.x-Y.unif}
Let $Y\sim\cU([0,\,1])$ and let $f^0\in\cP([0,\,1])$ be any initial awareness distribution with mean $m_X^0\in [0,\,1]$. If $f(t)\in\cP([0,\,1])$, $t>0$, is the solution to~\eqref{eq:Boltzmann.x} issuing from $f^0$ and $f^\infty\in\cP([0,\,1])$ is the Maxwellian defined by the Fourier transform~\eqref{eq:hatf.mX0} then, for every $s\in (0,\,2]$,
$$ d_s(f(t),f^\infty)\leq d_s(f^0,f^\infty)e^{-\left[1-(1-\alpha)^s\right]t}, \qquad t>0. $$
In particular,
$$ \lim_{t\to +\infty}d_s(f(t),f^\infty)=0 $$
for every $\alpha\in (0,\,1]$.
\end{theorem}

\begin{remark}
Theorem~\ref{theo:asymptotics.x-Y.unif} is stronger than Corollary~\ref{cor:asymptotics.x} in that the convergence of any solution of~\eqref{eq:Boltzmann.x} to the Maxwellian~\eqref{eq:hatf.mX0} holds with no restrictions on the parameter $\alpha$. Moreover, the convergence is achieved in $s$-Fourier metrics up to $s=2$, because in this case the statistical moments of the probability measures involved in the problem coincide up to the order $1$. Notice that the value of $s$ affects the speed of convergence of $f$ to $f^\infty$. In particular, the higher $s$ the faster the convergence.
\end{remark}

\begin{proof}[Proof of Theorem~\ref{theo:asymptotics.x-Y.unif}]
Since $f^\infty$ is a (constant-in-time) solution to~\eqref{eq:Boltzmann.x}, we may proceed like in the proof of Theorem~\ref{theo:ds} taking $f_1=f$ and $f_2=f^\infty$. We have then:
$$ \partial_t\frac{\hat{f}-\hat{f}^\infty}{\abs{\xi}^s}=H(m_X^0,\xi)\frac{\hat{f}((1-\alpha)\xi,t)-\hat{f}^\infty((1-\alpha)\xi)}{\abs{\xi}^s}-\frac{\hat{f}-\hat{f}^\infty}{\abs{\xi}^s}, $$
where it should be noticed that the term $H(m_X^0,\xi)$ at the right-hand side can now be collected because, due to $Y\sim\cU([0,\,1])$, $f(t)$ and $f^\infty$ have the same mean $m_X^0$ for all $t>0$. Consequently, letting $h(\xi,t):=\frac{\hat{f}(\xi,t)-\hat{f}^\infty(\xi)}{\abs{\xi}^s}$, we obtain:
$$ \partial_th+h=H(m_X^0,\xi)\frac{\hat{f}((1-\alpha)\xi,t)-\hat{f}^\infty((1-\alpha)\xi)}{\abs{\xi}^s}, $$
whence, multiplying both sides by $e^t$ and integrating in time on the interval $[0,\,t]$, $t>0$,
$$ e^th(\xi,t)=h_0(\xi)+H(m_X^0,\xi)\int_0^te^\tau\frac{\hat{f}((1-\alpha)\xi,\tau)-\hat{f}^\infty((1-\alpha)\xi)}{\abs{\xi}^s}\,d\tau, $$
where $h_0(\xi):=h(\xi,0)=\frac{\hat{f}^0(\xi)-\hat{f}^\infty(\xi)}{\abs{\xi}^s}$. Arguing like in the proof of Theorem~\ref{theo:ds}, we obtain then the following bound:
$$ e^t\norminf{h(t)}\leq\norminf{h_0}+(1-\alpha)^s\int_0^te^\tau\norminf{h(\tau)}\,d\tau $$
with $\norminf{h(t)}=d_s(f(t),f^\infty)$, whence the thesis follows by applying Gr\"onwall's inequality to the function $e^t\norminf{h(t)}$.
\end{proof}

\subsection{Numerical tests}
In this section, we show some numerical solutions of the Boltzmann-type equation~\eqref{eq:Boltzmann.x} with interaction rules defined by~\eqref{eq:x'}, which illustrate the previous theoretical results. To solve~\eqref{eq:Boltzmann.x} numerically we use a modification of the Nanbu-Babovski Monte Carlo algorithm, which is based on the implementation of discrete-in-time stochastic particle dynamics producing the kinetic equation~\eqref{eq:Boltzmann.x} in the continuous time limit, cf. e.g.,~\cite{fraia2020RUMI,pareschi2013BOOK} for details.

\begin{figure}[!t]
    \centering
    \subfigure[]{\includegraphics[width=.35\textwidth]{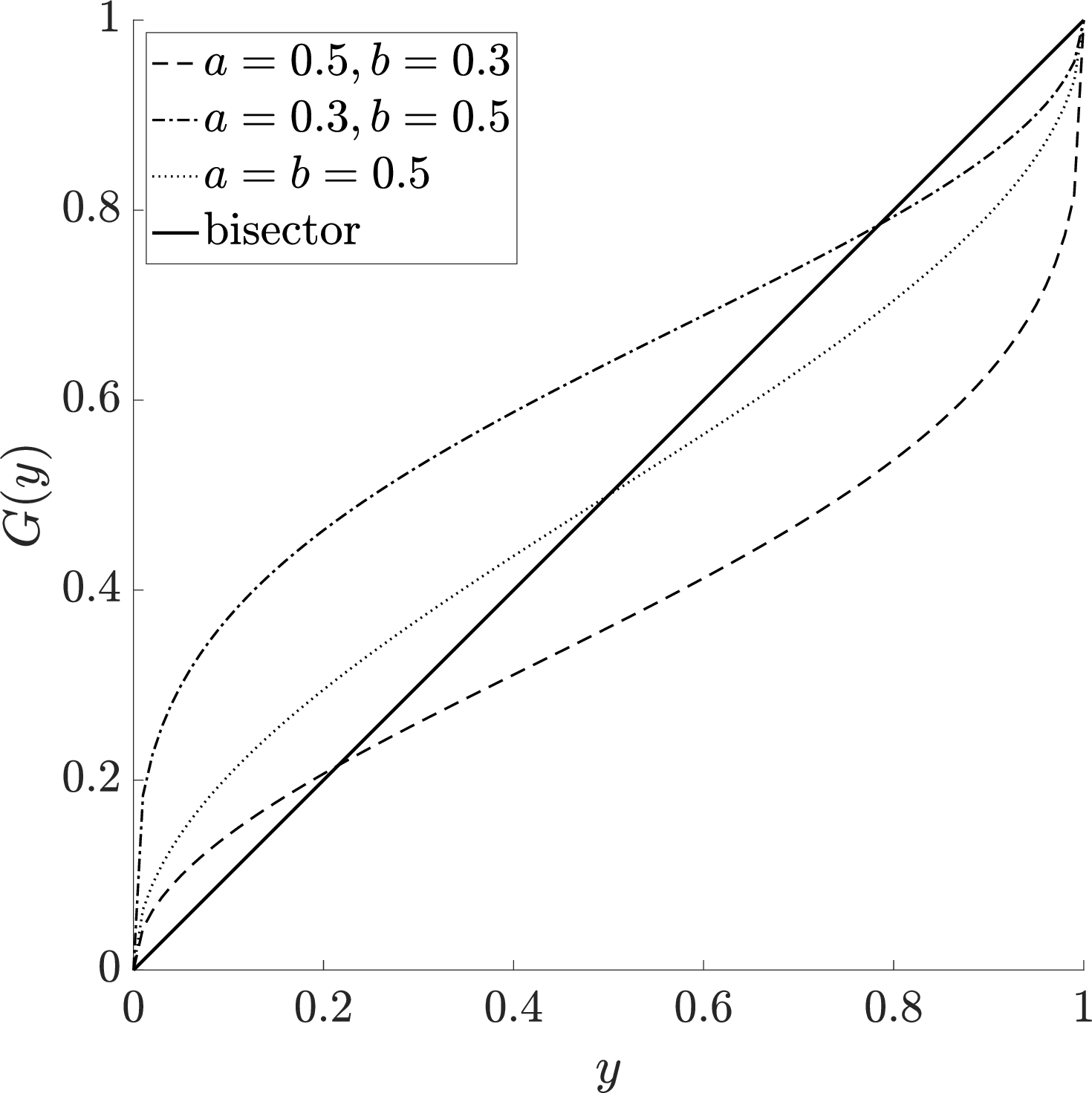} \label{fig:betacdf_stable}}
    \qquad
    \subfigure[]{\includegraphics[width=.35\textwidth]{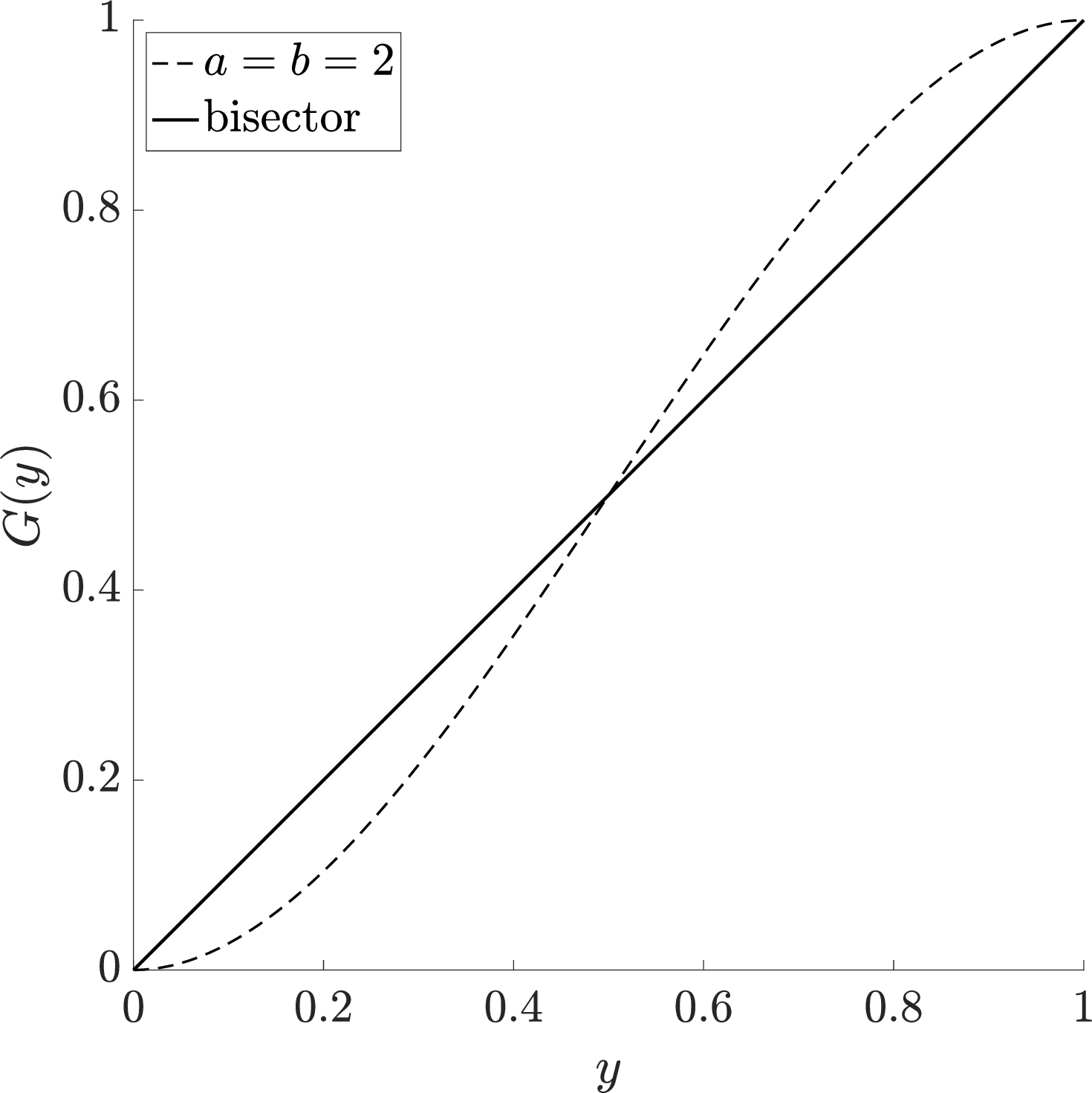} \label{fig:betacdf_no_stable}} 
    \caption{Cumulative distribution function of $Y\sim\operatorname{Beta}(a,\,b)$ for different values of the parameters $a,\,b$. Panel (a) shows three sets of parameters for which an asymptotically stable value of $m_X^\infty\in (0,\,1)$ exists. Panel (b) shows instead a set of parameters for which $m_X^\infty\in (0,\,1)$ is not asymptotically stable}
   \label{fig:betacdf}
\end{figure}

In the following tests we set $Y\sim\operatorname{Beta}(a,\,b)$, i.e. we fix
$$ g(y)=\frac{y^{a-1}(1-y)^{b-1}}{\operatorname{B}(a,b)}, \qquad y\in [0,\,1], \quad a,\,b\in\R_+, $$
where $\operatorname{B}(\cdot,\cdot)$ denotes the beta function. With this distribution, tuning the parameters $a,\,b$, we can simulate different scenarios of the evolution of the mean awareness $m_X$ at the basis of the large time trend of the whole system, cf. Section~\ref{sect:kinetic_description.x} and Figure~\ref{fig:betacdf}.

\begin{figure}[!t]
    \centering
    \subfigure[]{\includegraphics[width=.35\textwidth]{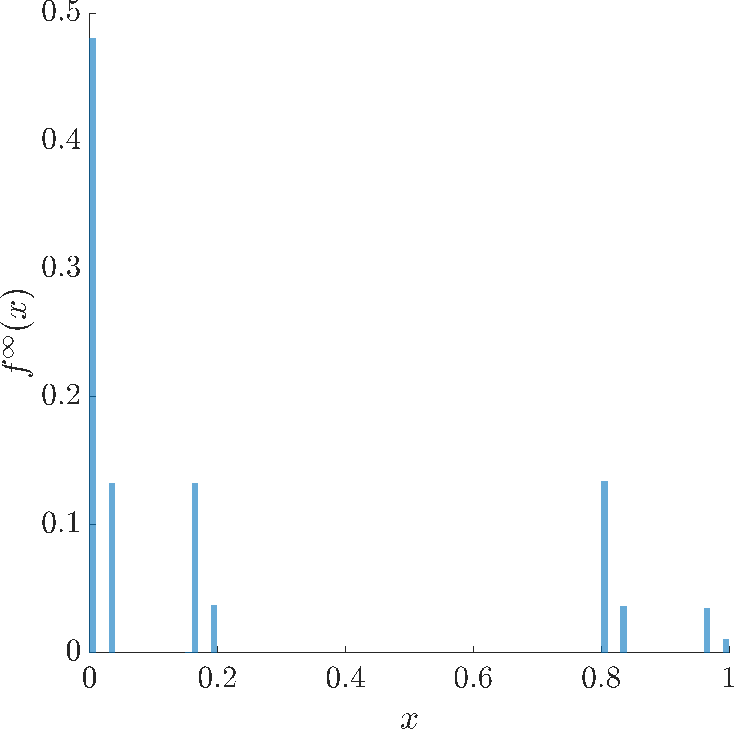} \label{fig:as_stabile_cluster}}
    \qquad
    \subfigure[]{\includegraphics[width=.35\textwidth]{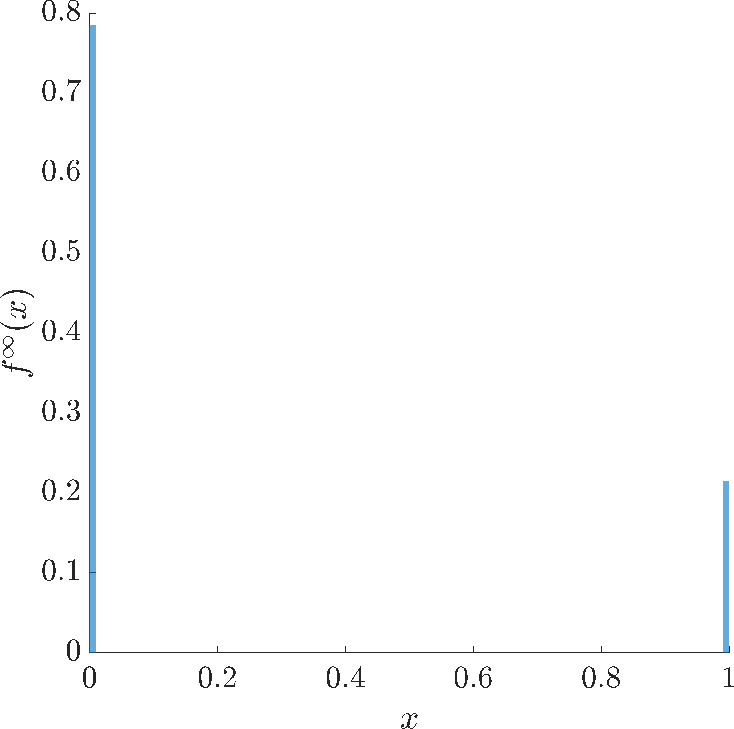} \label{fig:as_stabile_2peaks}}
    \caption{Asymptotic awareness distribution for $Y\sim\operatorname{Beta}(0.5,\,0.3)$ when $\alpha$ satisfies~\eqref{ass:alpha}. In particular: (a) $\alpha<1$, (b) $\alpha=1$}
    \label{fig:as_stabile}
\end{figure}

We begin by fixing $a=0.5$, $b=0.3$, that correspond to a case in which an asymptotically stable mean awareness $m_X^\infty\in (0,\,1)$ exists, cf. Figure~\ref{fig:betacdf_stable}. In particular, we compute numerically that $m_X^{\infty}\approx 0.215$. According to Corollary~\ref{cor:asymptotics.x}, if $\alpha$ satisfies~\eqref{ass:alpha}, which in this case yields $\alpha>0.695$ for $s=1$, the system converges in time to the Maxwellian univocally characterised by~\eqref{eq:hatf}. This is clearly shown in Figure~\ref{fig:as_stabile_cluster}, that we have obtained with $\alpha=0.8$: the asymptotic distribution of the awareness clusters in specific points with specific weights as predicted by the theory, cf.~\eqref{eq:x.infty}. In particular, we have numerically:
\begin{equation*}
	\begin{aligned}
		x^\infty_1 &= 0 \\
		x^\infty_2 &= \alpha=0.8 \\
		x^\infty_3 &= \alpha(1-\alpha)=0.16  \\
		x^\infty_4 &= \alpha(1-\alpha)^2=0.032 \\
		x^\infty_5 &= \alpha+\alpha(1-\alpha)=0.96 \\
		x^\infty_6 &= \alpha+\alpha(1-\alpha)^2=0.832\\
		x^\infty_7 &= \alpha(1-\alpha)+\alpha(1-\alpha)^2=0.192 \\
		x^\infty_8 &= \alpha+\alpha(1-\alpha)+\alpha(1-\alpha)^2=0.992
	\end{aligned}
	\begin{aligned}
		&\vphantom{
			\begin{aligned}
				x^\infty_1 &= 0
			\end{aligned}
		} \quad\quad\ \text{with weight } \left(1-G(m_X^\infty)\right)^3\approx 0.484 \\
		&\left.
		\vphantom{
			\begin{aligned}
				x^\infty_2 &= \alpha \\
				x^\infty_3 &= \alpha(1-\alpha) \\
				x^\infty_4 &= \alpha(1-\alpha)^2
			\end{aligned}
		}\right\} \quad \text{with weight } G(m_X^\infty)\left(1-G(m_X^\infty)\right)^2\approx 0.132 \\
		&\left.
		\vphantom{
			\begin{aligned}
				x^\infty_5 &= \alpha+\alpha(1-\alpha) \\
				x^\infty_6 &= \alpha+\alpha(1-\alpha)^2 \\
				x^\infty_7 &= \alpha(1-\alpha)+\alpha(1-\alpha)^2
			\end{aligned}
		}\right\} \quad \text{with weight } G^2(m_X^\infty)\left(1-G(m_X^\infty)\right)\approx 0.036 \\
		&\vphantom{
			\begin{aligned}
				x^\infty_8 &= \alpha+\alpha(1-\alpha)+\alpha(1-\alpha)^2
			\end{aligned}
		} \quad\quad\ \text{with weight } G^3(m_X^\infty)\approx 0.01.
	\end{aligned}
	\label{eq:x.infty.values}
\end{equation*}

\begin{remark}
In principle, the distribution featuring these clusters and weights should be an approximation of the actual Maxwellian $f^\infty$, since it is deduced by considering only the first three terms of an infinite convolution. Nevertheless, numerically one observes precisely this approximate Maxwellian with no additional clusters. The reason is that the terms $f_k^\infty$ for $k\geq 3$, which are neglected in the convolution, are so close to $\delta_0$ that, in the numerical approximation, they are indistinguishable from $\delta_0$. From~\eqref{eq:dist.fk-d0}, using the parameters of the simulation in Figure~\ref{fig:as_stabile_cluster}, i.e. $\alpha=0.8$ and $m_X^\infty=G(m_X^\infty)\approx 0.215$, we compute indeed:
$$ d_1(f_k^\infty,\delta_0)\leq d_1(f_3^\infty,\delta_0)\leq 0.215\cdot 0.8\cdot 0.2^3\approx 1.4\cdot 10^{-3}, \qquad \forall\,k\geq 3. $$
Hence, in the numerical simulation only the clusters arising from the convolution of $f_k^{\infty}$ for $k=0,\,1,\,2$ emerge.
\end{remark}

With $\alpha=1$ we obtain instead that the asymptotic awareness distribution clusters only in the two extreme points $x=0,\,1$ as portrayed by Figure~\ref{fig:as_stabile_2peaks} and confirmed analytically by~\eqref{eq:f^inf.x}.

\begin{figure}[!t]
    \centering
    \subfigure[]{\includegraphics[width=.35\textwidth]{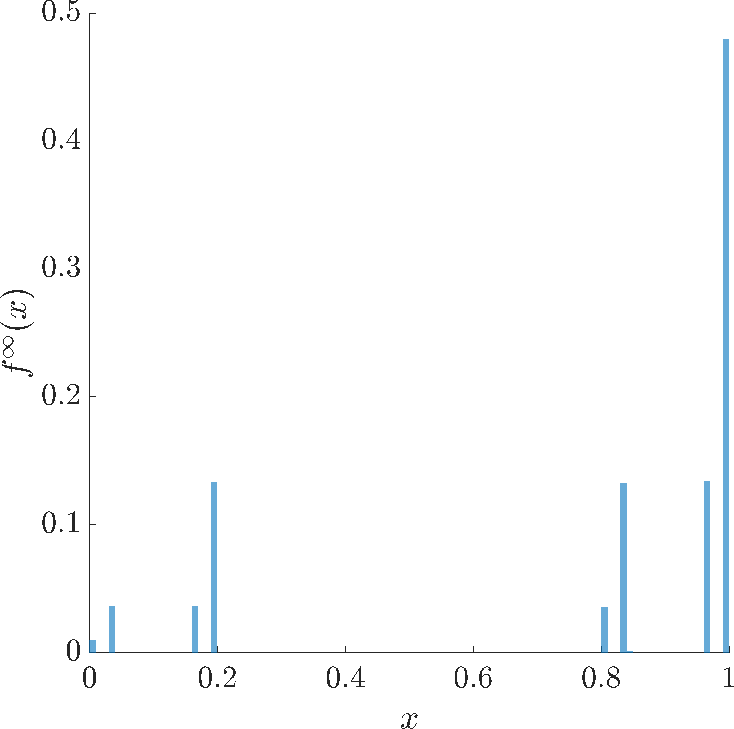} \label{fig:cluster_0.7}}
    \qquad
    \subfigure[]{\includegraphics[width=.35\textwidth]{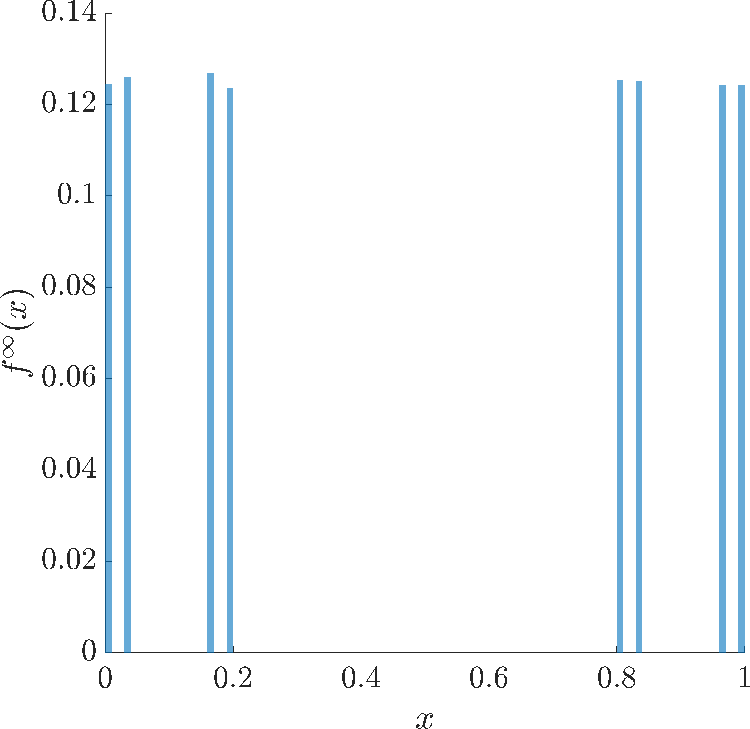} \label{fig:cluster_0.5}}
    \caption{Asymptotic awareness distribution for (a) $Y\sim\operatorname{Beta}(0.3,\,0.5)$ ($m_X^\infty=0.785$, cf. Figure~\ref{fig:betacdf_stable}) and (b) $Y\sim\operatorname{Beta}(0.5,\,0.5)$ ($m_X^\infty=0.5$, cf. Figure~\ref{fig:betacdf_stable})}
    \label{fig:cluster_diversa_h}
\end{figure}

Changing the parameters $a$, $b$ of the probability density function $g$ of $Y$ we experiment how the asymptotic awareness distribution is affected by different distributions of the reliability of news. In Figure~\ref{fig:cluster_diversa_h} we show the $f^\infty$'s computed numerically in correspondence of the other two choices of $a$, $b$ displayed in Figure~\ref{fig:betacdf_stable}. We notice, in particular, that $f^\infty$ still clusters in the same points as before, indeed from~\eqref{eq:x.infty} we see that such points are not affected by $g$. Nevertheless, the height of the peaks, viz. the percentage of individuals in the various awareness clusters, varies consistently with the dependence of the weights in~\eqref{eq:x.infty} on the cumulative distribution function $G$ of $Y$.

\begin{figure}[!t]
    \centering
    \subfigure[]{\includegraphics[width=.35\textwidth]{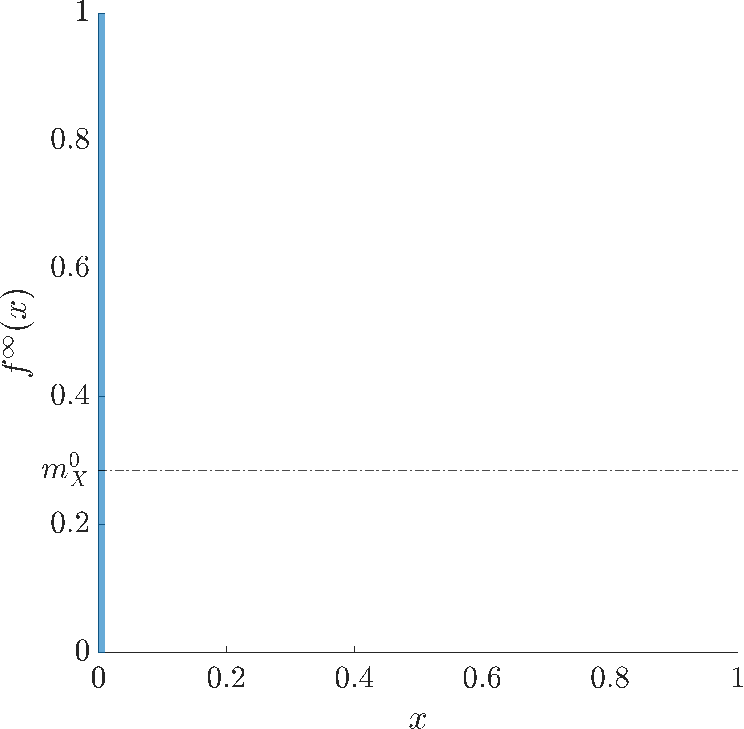} \label{fig:unstable_0}}
    \qquad
    \subfigure[]{\includegraphics[width=.35\textwidth]{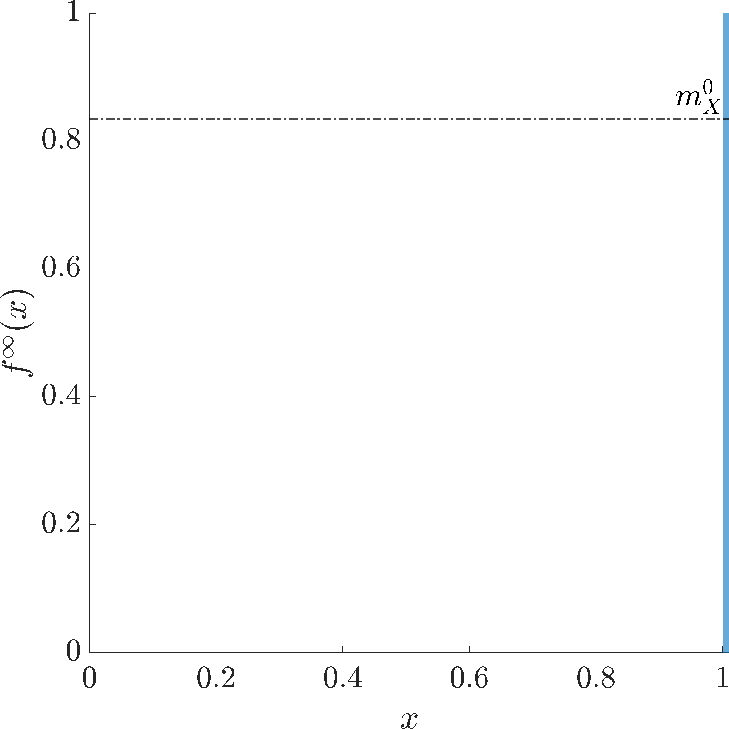} \label{fig:unstable_1}}
    \caption{Asymptotic awareness distribution for $Y\sim\operatorname{Beta}(2,\,2)$ ($m_X^\infty=0.5$ unstable, cf. Figure~\ref{fig:betacdf_no_stable}) and: (a) $m_X^0<0.5$; (b) $m_X^0>0.5$}
    \label{fig:unstable}
\end{figure}

In Figure~\ref{fig:unstable} we show instead the asymptotic awareness distribution reached when the distribution of $Y$ is such that $m_X^\infty\in (0,\,1)$ is unstable. This happens e.g., by fixing $a=b=2$, which produces $m_X^\infty=0.5$, cf. Figure~\ref{fig:betacdf_no_stable}. In this case, as predicted by the theory, $f^\infty$ clusters on either $x=0$ if $m_X^0<0.5$, cf. Figure~\ref{fig:unstable_0}, or $x=1$ if $m_X^0>0.5$, cf. Figure~\ref{fig:unstable_1}.

\subsection{Inferring the reliability of news from real data}
Bringing the concept of reliability of news to a quantitative basis, by assigning to it numerical values on a scale from $0$ to $1$, may appear an abstraction for speculative theoretical purposes. Instead, in this section we suggest a possible conceptual method to estimate such a fundamental input from real data. We stress however that we are not going to elaborate on real data. On the contrary, we assume that a database of news is available, whose contents have already been flagged as ``fake'' or ``non-fake'' according to some classification criterion that we do not discuss here.

It is reasonable to argue that the probability of identifying a fake news depends significantly on the news topic, due to various factors linked to the nature and context of the disseminated information. Several arguments support this intuition: for instance, trending topics or topics of high public interest (such as e.g., political events, celebrity news, health crises) are more likely to attract fake information. This is because they have a higher potential for virality, which makes them attractive targets for spreading misinformation to gain attention or to influence the public opinion. Conversely, less discussed topics may suffer from fewer fake news because misinformation campaigns are less motivated to focus on them. In case of topics requiring specialised knowledge, such as e.g., scientific research or medical advice, fake information may be more prevalent and harder to detect as the general public may lack the expertise to discern factual accuracy. Moreover, topics such as disasters, scandals, and controversial issues are particularly prone to fake news because emotional engagement can override critical thinking. In view of all of this, we aim to define the reliability of news taking into account that the probability of bumping into fake news is intimately correlated to the news topic.

\emph{Clustering} is a powerful tool for identifying topic clusters, due to its ability to handle large, high-dimensional datasets, discover hidden patterns, and provide a method for grouping similar news based on the content. It is not our purpose to go into the details of the algorithms typically used for text clustering, as many books exist delving into data mining with applications to text data as well, see e.g.,~\cite{DataMiningBOOK,SocialMiningBOOK}. For completeness, we confine ourselves to outlining the main conceptual steps at the basis of the definition of ``proximity'' of textual data:
\begin{itemize}
\item \textit{Text preprocessing}: The first step in text clustering is preprocessing raw text data. This involves tasks such as breaking text into words or tokens (tokenisation), removing stop words (common words such as ``the'', ``and'', ``is'', that do not carry significant information), and reducing words
to their root form (stemming).
\item \textit{TF-IDF vectorisation}: Once the text is preprocessed, each news is transformed into a numerical representation using the Term Frequency-Inverse
Document Frequency (TF-IDF) method. TF-IDF is a statistical measure used to evaluate the importance of a word in a text relative to a collection of texts. The TF component measures the frequency of a word in a text, while the IDF component measures the inverse frequency of the word across all texts. The resulting TF-IDF scores reflect how significant a word is in a particular text while mitigating the influence of frequently occurring common words.
\item \textit{K-means clustering}: After transforming the text data into TF-IDF vectors, the $k$-means clustering algorithm is applied, which partitions the data into $k$ (fixed) clusters and assigns each news to the cluster with the nearest mean (centroid). The algorithm proceeds as follows:
\begin{itemize}
\item \textit{Initialisation}: select $k$ initial centroids randomly from the dataset;
\item \textit{Assignment}: assign each news to the nearest centroid based on the chosen distance (e.g., the Euclidean distance);
\item \textit{Update}: recalculate the centroids as the mean of all news assigned to each cluster;
\item \textit{Convergence}: repeat the ``Assignment'' and ``Update'' steps until the centroids no longer change significantly or a maximum number of iterations
is reached.
\end{itemize}
\end{itemize}
Finally, the topic of each cluster can be inferred by examining the words with higher value withing the cluster.

With this cluster identification, to each news about a certain topic (viz. within a certain cluster) we may associate a \textit{value of reliability} in the range $[0,\,1]$ defined as
$$ y:=\Prob(\text{true}\cap\text{topic})=\Prob(\text{true}\vert\text{topic})\cdot\Prob(\text{topic}), $$
where:
\begin{enumerate}[label=(\roman*)]
\item $\Prob(\text{true}\vert\text{topic})$ is the probability that the news is true conditional to the chosen topic (viz. cluster); it can be estimated empirically as the ratio between the number of non-fake news in the cluster and the total number of news in the same cluster;
\item $\Prob(\text{topic})$ is the probability to run into that topic, which can be estimated empirically as the ratio between the number of news in the cluster and the total number of news in all clusters.
\end{enumerate}

Since the clustering procedure identifies $k$ clusters, in this way we generate $k$ different values of reliability $y_1,\,y_2,\,\dots,\,y_k\in [0,\,1]$. From them, we estimate the probability distribution of the random variable $Y$ as
$$ \Prob(Y=y):=\sum_{\text{topic}\,:\,\Prob(\text{true}\cap\text{topic})=y}\Prob(\text{topic}), \qquad
    y\in\{y_1,\,y_2,\,\dots,\,y_k\}. $$

As mentioned at the beginning, this procedure assumes that news in the database are already classified in ``true'' or ``false'' according to some criterion, so that, after clustering, it is possible to evaluate $\Prob(\text{true}\vert\text{topic})$ straightforwardly. Nevertheless, we stress that such a classification is a by far non-trivial further issue because the definition itself of ``fake news'' is inevitably partly ambiguous. 

Moreover, we observe that a \textit{uniform} trend of $Y$, apart from being pleasant for speculative theoretical purposes (cf. Section~\ref{sect:Y.unif}), may be expected qualitatively from the proposed method in realistic scenarios. For instance, in the case of datasets collecting news from different sources over a short period of time; or in the case of datasets centred around a single topic due to e.g., filters on the keywords.

\section{Rise and fall of the popularity of fake news}
\label{sect:popularity}
The dissemination of fake news depends on the ability/inability of the individuals to recognise them as unreliable pieces of information and to decide consequently whether to share them or not with the others. In this section, we propose a simple model to investigate the trends of the popularity of fake news arising from exchanges among the individuals such as those taking place on social networks. Like in Section~\ref{sect:awareness}, we consider interactions between agents and news but in this case the agents have to be understood as a background for the spread of news depending on their inclination to repost more or less reliable news. In this context, agents are not only characterised by their awareness $x\in [0,\,1]$ but also by a variable $c\in\R_+$ representing their \textit{connectivity}, viz. a measure of their number of contacts in the social network. On the other hand, news is not only characterised by the reliability $y\in [0,\,1]$ but also by the \textit{degree of popularity} $v\in\R_+$, whose evolution in time we aim to investigate statistically.

We assume that the degree of popularity evolves according to the following elementary rule:
\begin{equation}
	v'=(1-\mu)v+P(x,y,c).
	\label{eq:v'}
\end{equation}
In this formula, inspired by~\cite{toscani2018PRE}, $\mu\in (0,\,1)$ is the physiological decay rate of the popularity of news whereas $P$ is the posting function, which determines an increase in the popularity of a certain content if the latter is reposted by the users of the social network. In more detail, $P$ depends on the awareness $x$ of an individual interacting with news with reliability $y$ and on the connectivity $c$ of that individual. Specifically, we set:
\begin{equation}
	P(x,y,c)=\nu c\chi(x\leq y+\beta(x)),
	\label{eq:P}
\end{equation}
where $\beta:[0,\,1]\to\R_+$ is a function expressing the gap between the awareness of the individual and the reliability of news which induces the former to repost a content even if the content is not perceived as completely trustworthy. Notice indeed that~\eqref{eq:P} allows for $P\neq 0$ whenever $y\geq x-\beta(x)$, thus in particular with a reliability $y$ possibly strictly less than the individual awareness $x$. In other words, $\beta(x)$ can be understood as the propensity of an individual with awareness $x$ to consciously trust news which they know to be not completely true. For technical purposes, we assume that the function $x-\beta(x)$ is invertible in $[0,\,1]$, which is the case if e.g., $\beta$ is non-increasing in $[0,\,1]$. Notice that this is a reasonable assumption also from the modelling point of view.

Finally, in~\eqref{eq:P} the increase in the popularity of a reposted content is $\nu c$, where $\nu>0$ is a parameter, i.e., it is proportional to the connectivity of the reposting individual. The rationale is that the higher the number of connections of an individual the larger the pool of social network users reached by reposted news.

\begin{proposition} \label{prop:v'}
Rule~\eqref{eq:v'}-\eqref{eq:P} is physically admissible for every $\mu\in (0,\,1)$, $\nu>0$, $\beta\geq 0$, i.e. $v'\in\R_+$ for all $v,\,c\in\R_+$ and all $x,\,y\in [0,\,1]$.
\end{proposition}
\begin{proof}
Straightforward, observing that $v'$ is the sum of non-negative terms.
\end{proof}

\subsection{Kinetic description and trend to equilibrium}
Let $p=p(v,y,t):\R_+\times [0,\,1]\times [0,\,+\infty)\to\R_+$ be the probability density function of news with popularity $v$ and reliability $y$ at time $t$. Invoking the same principles from~\cite{pareschi2013BOOK}, which already led to the kinetic equation~\eqref{eq:Boltzmann.x}, we write a Boltzmann-type equation in weak form for the evolution of $p$ under the interaction rule~\eqref{eq:v'}-\eqref{eq:P}:
\begin{multline}
	\frac{d}{dt}\int_0^1\int_0^{+\infty}\Phi(v,y)p(v,y,t)\,dv\,dy \\
		=\int_0^{+\infty}\int_0^1\int_0^1\int_0^{+\infty}\left(\Phi(v',y)-\Phi(v,y)\right)p(v,y,t)f(x,t)C(c)\,dv\,dy\,dx\,dc,
	\label{eq:Boltzmann.p}
\end{multline}
where $\Phi:\R_+\times [0,\,1]\to\C$ is an arbitrary observable quantity (test function). Proposition~\ref{prop:v'} ensures that the term $\phi(v',c)$ is well-defined and that
$$ \supp{p(\cdot,\cdot,t)}\subseteq\R_+\times [0,\,1] $$
for all $t>0$ if it is so at $t=0$.

In~\eqref{eq:Boltzmann.p}, the probability density function of the users of the social network who interact with the news is $f(x,t)C(c)$, where we are assuming statistical independence between the awareness $x \in [0,1]$ and the connectivity $c \in \R_+$ at each time $t\geq 0$. Specifically, $C\in\cP(\R_+)$ is the probability distribution of the connectivity of the users of the social network. Consistently, we assume the normalisation condition $\int_0^{+\infty}C(c)\,dc=1$; moreover, we denote by
$$ m_C:=\int_0^{+\infty}cC(c)\,dc $$
the mean connectivity of the users of the social network and by
$$ E_C:=\int_0^{+\infty}c^2C(c)\,dc $$
the energy (second moment) of the connectivity distribution. Notice that we are implicitly considering a \textit{static} social network, i.e., one in which the connections among the individuals do not change in time, because $C$ is independent of $t$. By fixing a statistical big picture of the network topology, this simplification allows us to investigate the impact of networked interactions on the spread of fake news regardless of inessential local network rewiring.

From~\eqref{eq:Boltzmann.p} we may easily obtain an evolution equation for the conditional probability distribution of the popularity $V$ given the reliability $Y$ of news, say $p_y=p_y(v,t):\R_+\times [0,\,+\infty)\to\R_+$, by invoking the disintegration theorem of a measure. Specifically, writing
$$ p(v,y,t)=p_y(v,t)\otimes g(y), $$
where $g$ is the probability distribution of $Y$ introduced in Section~\ref{sect:awareness}, and plugging this into~\eqref{eq:Boltzmann.p} with the choice $\Phi(v,y)=\phi(v)\psi(y)$ we get, owing to the arbitrariness of $\psi$,
$$ \frac{d}{dt}\int_0^{+\infty}\phi(v)p_y(v,t)\,dv=\int_0^{+\infty}\int_0^1\int_0^{+\infty}\left(\phi(v')-\phi(v)\right)p_y(v,t)f(x,t)C(c)\,dv\,dx\,dc. $$
Unlike $p$, the conditional probability distribution $p_y$ provides a closer perspective on possibly different trends of the popularity depending on the various levels of reliability of news. Therefore, from now on we will focus invariably on $p_y$.

As a further simplification, we assume that the awareness distribution $f$ may be replaced by its asymptotic profile $f^\infty$, which amounts to considering the learning process of Section~\ref{sect:awareness} much quicker than the reposting dynamics which shape the popularity of news. In other words, we imagine that the reposting of news takes place on an already consolidated background of awareness of the social network users. This may not be always true in practice but such an approximation allows for a deeper analytical understanding of the model while being certainly reasonable at least in selected scenarios. In particular, among all possible awareness equilibrium distributions discussed in Section~\ref{sect:Maxwellian.x}, we stick to the one corresponding to $\alpha=1$ as it can be given an explicit analytical representation, cf.~\eqref{eq:f^inf.x}.

On the whole, we consider therefore the following weak Boltzmann-type equation for $p_y$:
\begin{equation}
	\frac{d}{dt}\int_0^{+\infty}\phi(v)p_y(v,t)\,dv=\int_0^{+\infty}\int_0^1\int_0^{+\infty}\left(\phi(v')-\phi(v)\right)p_y(v,t)f^\infty(x)C(c)\,dv\,dx\,dc,
	\label{eq:Boltzmann.p_y}
\end{equation}
where $\phi:\R_+\to\C$ is arbitrary and $f^\infty=\left(1-G(m_X^\infty)\right)\delta_0+G(m_X^\infty)\delta_1$. By means of techniques analogous to those employed in Theorem~\ref{theo:exist_uniq}, cf. Appendix~\ref{sect:proof} and~\cite{carrillo2007RMUP}, it is possible to prove that also~\eqref{eq:Boltzmann.p_y} admits a unique solution $p_y(t)=p_y(\cdot,t)\in\cP(\R_+)$, $t>0$, in correspondence of any initial condition $p_y^0\in\cP(\R_+)$.

With $\phi(v)=v$ we study the time evolution of the mean popularity of news with reliability $y$:
$$ m_{V\vert y}(t):=\int_0^{+\infty}vp_y(v,t)\,dv. $$
Specifically, from~\eqref{eq:Boltzmann.p_y} and taking~\eqref{eq:v'},~\eqref{eq:P} into account we get
$$ \dot{m}_{V\vert y}=-\mu m_{V\vert y}+\nu m_C\bigl(1-G(m_X^\infty)\chi(y<1-\beta(1))\bigr), $$
whence
$$ m_{V\vert y}(t)=e^{-\mu t}m_{V\vert y}^0+\frac{\nu m_C\bigl(1-G(m_X^\infty)\chi(y<1-\beta(1))\bigr)}{\mu}\left(1-e^{-\mu t}\right). $$
Therefore, we see that $m_{V\vert y}$ converges exponentially fast to
\begin{equation}
	m_{V\vert y}^\infty:=
	\begin{cases}
		\dfrac{\nu}{\mu}m_C\bigl(1-G(m_X^\infty)\bigr) & \text{if } y<1-\beta(1) \\[5mm]
		\dfrac{\nu}{\mu}m_C & \text{if } y\geq 1-\beta(1)
	\end{cases}
	\label{eq:m_V|y.inf}
\end{equation}
when $t\to +\infty$. We notice that only in the case of poorly reliable news, i.e. for $y<1-\beta(1)$, the asymptotic mean popularity $m_{V\vert y}^\infty$ depends on the mean awareness $m_X^\infty$ of the population. On the contrary, in the case of sufficiently reliable news, i.e., for $y\geq 1-\beta(1)$, $m_{V\vert y}^\infty$ is independent of $m_X^\infty$. Moreover, according to this model poorly reliable news reaches always a lower asymptotic popularity than sufficiently reliable news, indeed $1-G(m_X^\infty)<1$. Nevertheless, while in the ideal regime $\beta\equiv 0$, i.e., when no one reposts a content recognised as unreliable, only fully reliable news ($y=1$) attain the maximum asymptotic popularity $\frac{\nu}{\mu}m_C$, for $\beta\not\equiv 0$ also part of non-completely reliable news reaches, in the long run, the maximum popularity. In all cases, $m_{V\vert y}^\infty$ is proportional to the mean connectivity $m_C$ of the individuals, which shows explicitly the impact of the background network on the average trend of the popularity.

\begin{remark} In~\eqref{eq:m_V|y.inf}, the watershed between poorly and sufficiently reliable news turns out to be the value $1-\beta(1)$ of the reliability $y$. This value has a meaningful modelling interpretation: since $\beta(1)$ is the level of falseness that a fully aware individual, i.e. one with awareness $x=1$, agrees to tolerate when reposting contents, any news with reliability at least $1-\beta(1)$ is necessarily indistinguishable from a completely true content, i.e., one with $y=1$.
\end{remark}

Letting now $\phi(v)=v^2$ in~\eqref{eq:Boltzmann.p_y} we study the energy of the popularity of news with reliability $y$:
$$ E_{V\vert y}(t):=\int_0^{+\infty}v^2p_y(v,t)\,dv. $$
In particular, we obtain
$$ \dot{E}_{V\vert y}=-\mu(2-\mu)E_{V\vert y}+\nu\left(2(1-\mu)m_Cm_{V\vert y}+\nu E_C\right)\left(1-G(m_X^\infty)\chi(y<1-\beta(1))\right), $$
which, since $m_{V\vert y}\to m_{V\vert y}^\infty$ exponentially fast for $t\to +\infty$, implies that $E_{V\vert y}$ converges to
$$ E_{V\vert y}^\infty:=\frac{\nu}{\mu(2-\mu)}\left(2(1-\mu)m_Cm_{V\vert y}^\infty+\nu E_C\right)\left(1-G(m_X^\infty)\chi(y<1-\beta(1))\right) $$
when $t\to +\infty$. More explicitly, using~\eqref{eq:m_V|y.inf} we write
$$ E_{V\vert y}^\infty=
	\begin{cases}
		\dfrac{\nu^2}{\mu(2-\mu)}\left(\dfrac{2(1-\mu)}{\mu}m_C^2(1-G(m_X^\infty))+E_C\right)(1-G(m_X^\infty)) & \text{if } y<1-\beta(1) \\[5mm]
		\dfrac{\nu^2}{\mu(2-\mu)}\left(\dfrac{2(1-\mu)}{\mu}m_C^2+E_C\right) & \text{if } y\geq 1-\beta(1), \\
	\end{cases} $$
whence we observe that also the asymptotic second moment of the conditional popularity distribution $p_y$ depends on the mean awareness $m_X^\infty$ of the population only in the case of poorly reliable news.

Interestingly, by computing the asymptotic variance
$$ \sigma^{2,\infty}_{V\vert y}:=E_{V\vert y}^\infty-{(m_{V\vert y}^\infty)}^2 $$
of the conditional popularity distribution we discover
$$ \sigma^{2,\infty}_{V\vert y}=
	\begin{cases}
		\dfrac{\nu^2(1-G(m_X^\infty))}{\mu(2-\mu)}\left(E_C-(1-G(m_X^\infty))m_C^2\right) & \text{if } y<1-\beta(1) \\[5mm]
		\dfrac{\nu^2}{\mu(2-\mu)}\left(E_C-m_C^2\right) & \text{if } y\geq 1-\beta(1),
	\end{cases} $$
which, noting that $\sigma^2_C:=E_C-m_C^2$ is the variance of the connectivity distribution of the social network, implies that for $y\geq 1-\beta(1)$, i.e. in the case of sufficiently reliable news, $\sigma^{2,\infty}_{V\vert y}$ is proportional to $\sigma^2_C$, whereas for $y<1-\beta(1)$, i.e. in the case of poorly reliable news, $\sigma^{2,\infty}_{V\vert y}$ is bounded below by $\sigma^2_C$ as
$$ \sigma^{2,\infty}_{V\vert y}\geq\frac{\nu^2(1-G(m_X^\infty))}{\mu(2-\mu)}\sigma^2_C. $$
Therefore, independently of the reliability $y$, the heterogeneity of the connectivity distribution of the social network shapes the variability of the popularity asymptotically reached by news. This result stresses once again explicitly the impact of the background network on the emergent statistics of the popularity of contents shared by the individuals.

In particular, if $\sigma^2_C>0$, i.e. if the connectivity is non-constant among the individuals, then also $\sigma^{2,\infty}_{V\vert y}>0$ for every $y\in [0,\,1]$. This indicates that the process of popularity formation modelled by~\eqref{eq:v'},~\eqref{eq:P} may give rise to non-trivial asymptotic distributions.

\subsubsection{Identification of the Maxwellian}
Proceeding like in Section~\ref{sect:Maxwellian.x}, we may characterise completely the equilibrium distribution of the popularity $p_y^\infty$ by switching to the Fourier-transformed version of~\eqref{eq:Boltzmann.p_y}. With $\phi(v)=e^{-i\xi v}$ we obtain, in particular,
$$ \partial_t\hat{p}_y(\xi,t)=\left[(1-G(m_X^\infty))\hat{C}(\nu\xi)+G(m_X^\infty)\hat{C}(\nu\chi(y\geq 1-\beta(1))\xi)\right]\hat{p}_y((1-\mu)\xi,t)-\hat{p}_y(\xi,t), $$
which, letting
\begin{align}
	\begin{aligned}[b]
	K_y(\xi) &:= (1-G(m_X^\infty))\hat{C}(\nu\xi)+G(m_X^\infty)\hat{C}(\nu\chi(y\geq 1-\beta(1))\xi) \\
	&\phantom{:}=
		\begin{cases}
			(1-G(m_X^\infty))\hat{C}(\nu\xi)+G(m_X^\infty) & \text{if } y<1-\beta(1) \\
			\hat{C}(\nu\xi) & \text{if } y\geq 1-\beta(1),
		\end{cases}
	\end{aligned}
	\label{eq:K_y}
\end{align}
we rewrite compactly as
\begin{equation}
	\partial_t\hat{p}_y(\xi,t)=K_y(\xi)\hat{p}_y((1-\mu)\xi,t)-\hat{p}_y(\xi,t).
	\label{eq:Fourier.p_y}
\end{equation}
The stationary distribution $p_y^\infty$ fulfils then the recursive relationship
$$ \hat{p}_y^\infty(\xi)=K_y(\xi)\hat{p}_y^\infty((1-\mu)\xi), $$
whence, by an argument analogous to the one developed in the proof of Proposition~\ref{prop:Maxwellian.x}, we deduce
\begin{equation}
	\hat{p}_y^\infty(\xi)=\prod_{k=0}^{\infty}K_y((1-\mu)^k\xi).
	\label{eq:hatp}
\end{equation}

Since $\vert\hat{C}(\xi)\vert\leq 1$ for all $\xi\in\R$ (as a general property of the Fourier transform of a probability distribution), it results also $\abs{K_y(\xi)}\leq 1$ for all $\xi\in\R$, whence
$$ \abs{\hat{p}_y^\infty(\xi)}\leq\abs{K_y(\xi)}=\vert\hat{C}(\nu\xi)\vert $$
for $y\geq 1-\beta(1)$. Consequently,
\begin{align*}
	\normHm{p_y^\infty}{m}^2 &= \int_{\R}{(1+\xi^2)}^m\abs{\hat{p}_y^\infty(\xi)}^2\,d\xi\leq\int_{\R}{(1+\xi^2)}^m\vert\hat{C}(\nu\xi)\vert^2\,d\xi \\
	&\leq \frac{1}{\nu^{2m+1}}\int_{\R}{(1+\xi^2)}^m\vert\hat{C}(\xi)\vert^2\,d\xi \\
	&= \frac{1}{\nu^{2m+1}}\normHm{C}{m}^2, \qquad m\in\N,
\end{align*}
which says $p_y^\infty\in H^m(\R_+)$ whenever $C\in H^m(\R_+)$. In other words, for $y\geq 1-\beta(1)$, i.e. for sufficiently reliable news, a smooth connectivity distribution entails an analogously smooth asymptotic distribution of the popularity in the sense of Sobolev regularity. The same may instead not be true in general for $y<1-\beta(1)$, i.e. for poorly reliable news, as in this case $K_y$ is in general not proportional to $\hat{C}$.

From~\eqref{eq:K_y},~\eqref{eq:hatp} we infer that $p_y^\infty$ depends on $m_X^\infty$ only in the case of poorly reliable news ($y<1-\beta(1)$), whereas for sufficiently reliable news ($y\geq 1-\beta(1)$) it is independent of it. This consolidates in more generality what we already observed about the asymptotic mean popularity and the variance of the conditional popularity distribution. The physical interpretation of this fact is clear: contents which, on the whole, are considered reliable enough reach statistically a popularity independent of the collective social awareness against fake news. Conversely, the statistical profile of the popularity of debated contents is affected by the collective ability to recognise such contents as poorly reliable. For instance, in the limit case $m_X^\infty=0$, which describes a society completely prone to fake news, it results $G(m_X^\infty)=0$ and hence $K_y(\xi)=\hat{C}(\nu\xi)$ also for $y<1-\beta(1)$, therefore poorly reliable contents reach asymptotically the same popularity distribution as sufficiently reliable ones. Conversely, in the opposite limit case $m_X^\infty=1$, which depicts a society well immunised against fake news, it results $G(m_X^\infty)=1$. Consequently, for $y<1-\beta(1)$ we have $K_y(\xi)=1$, whence, owing to~\eqref{eq:hatp}, $\hat{p}_y^\infty(\xi)=1$, i.e. $p_y^\infty=\delta_0$. Therefore, in this case poorly reliable contents are spontaneously discarded by the society in the long run.

We now prove that the Maxwellian given by~\eqref{eq:hatp} is the unique equilibrium distribution of~\eqref{eq:Boltzmann.p_y}.
\begin{theorem}
\label{theo:asymptotics.v_y}
Let $y\in [0,\,1]$ be fixed. Any two solutions $p_y(t),\,q_y(t)\in\cP(\R_+)$ to~\eqref{eq:Boltzmann.p_y} issuing from respective initial conditions $p_y^0,\,q_y^0\in\cP(\R_+)$ are such that
$$ d_s(p_y(t),q_y(t))\leq d_s(p_y^0,q_y^0)e^{-[1-(1-\mu)^s]t}, \qquad t>0, $$
for every $s\in (0,\,1]$. Thus, in particular,
$$ \lim_{t\to +\infty}d_s(p_y(t),q_y(t))=0. $$
\end{theorem}
\begin{proof}
Taking the difference between the Fourier-transformed versions of~\eqref{eq:Boltzmann.p_y}, cf.~\eqref{eq:Fourier.p_y}, satisfied by $\hat{p}_y$ and $\hat{q}_y$ and dividing by $\abs{\xi}^s$ we get
$$ \partial_t\frac{\hat{p}_y-\hat{q}_y}{\abs{\xi}^s}=K_y(\xi)\frac{\hat{p}_y((1-\mu)\xi,t)-\hat{q}_y((1-\mu)\xi,t)}{\abs{\xi}^s}
	-\frac{\hat{p}_y-\hat{q}_y}{\abs{\xi}^s}. $$
At this point, the thesis follows arguing like in the proof of Theorem~\ref{theo:asymptotics.x-Y.unif} and considering additionally that, as already noticed, $\abs{K_y(\xi)}\leq 1$ for all $\xi\in\R$.
\end{proof}

\begin{corollary}
For every fixed $y\in [0,\,1]$, every solution $p_y(t)\in\cP(\R_+)$ to~\eqref{eq:Boltzmann.p_y} converges asymptotically in time to the Maxwellian defined by the Fourier transform~\eqref{eq:hatp}.
\end{corollary}
\begin{proof}
It suffices to take $q_y(v,t)=p_y^\infty(v)$, which is a constant-in-time solution to~\eqref{eq:Boltzmann.p_y}, in Theorem~\ref{theo:asymptotics.v_y}.
\end{proof}

\subsubsection{Popularity tails}
Since $\supp{p_y(\cdot,t)}\subseteq\R_+$, the question of the characterisation of the tail of $p_y$ arises. In particular, it is interesting to establish whether a \textit{fat tail} forms in $p_y$. We recall that $p_y$ is said to be \textit{fat-tailed} if, for $t>0$ fixed and $w>0$ large, there exists $r>0$ such that
\begin{equation}
	\int_w^{+\infty}p_y(v,t)\,dv\sim w^{-r}.
	\label{eq:Pareto_index}
\end{equation}
The exponent $r$ is called the \textit{Pareto index} of the distribution from the name of the Italian economist Vilfredo Pareto, who, at the beginning of the 20th century, observed empirically the polynomial decay of the tail of wealth distribution curves in western societies. More recently, several studies based on mathematical tools affine to those of the present work succeeded in explaining theoretically Pareto's observations. They also established precise analytical conditions linking the formation of a fat tail to the microscopic characteristics of the trading, see e.g.,~\cite{bisi2009CMS,cordier2005JSP,duering2009RMUP,gualandi2018ECONOMICS} and references therein. Notice that the smaller the Pareto index of a distribution the fatter the tail of the latter.

In the context of the popularity of news, a fat tail of $p_y$ indicates a non-negligible probability that a content with reliability $y$ becomes highly popular in time, cf.~\cite{toscani2018PRE}. ``Non-negligible'' has to be meant in comparison with the typical behaviour of systems of classical physics, whose statistical distributions decay to zero in general exponentially.

From~\eqref{eq:Pareto_index} we notice that if, at a certain time $t>0$, $p_y$ has a fat tail with Pareto index $r$ then $p_y(v,t)\sim v^{-(r+1)}$ when $v\to +\infty$. Consequently, from order $r$ onwards the statistical moments of $p_y$ blow to $+\infty$. The non-finiteness of some statistical moments provides an effective criterion to identify the formation of a fat tail in $p_y$.

To investigate this issue it is convenient to refer again to the Fourier-transformed version~\eqref{eq:Fourier.p_y} of~\eqref{eq:Boltzmann.p_y}, recalling the following relationship between the generic moment $M_n$ of order $n\in\N$ of $p_y$ and the Fourier transform $\hat{p}_y$:
$$ M_y^{(n)}(t):=\int_0^{+\infty}v^np_y(v,t)\,dv=i^n\partial_\xi^n\hat{p}_y(0,t). $$
Taking the $n$-th order $\xi$-derivative of~\eqref{eq:Fourier.p_y} and applying the Leibniz rule to the first term on the right-hand side we obtain
$$ \partial_t\partial_\xi^n\hat{p}_y(\xi,t)=\sum_{k=0}^{n}\binom{n}{k}\partial_\xi^{n-k}K_y(\xi)(1-\mu)^k\partial_\xi^k\hat{p}_y((1-\mu)\xi,t)
	-\partial_\xi^n\hat{p}_y(\xi,t). $$
Next, multiplying both sides by $i^n$ and evaluating in $\xi=0$ we discover
\begin{align*}
	\frac{dM_y^{(n)}}{dt} &= \sum_{k=0}^{n}\binom{n}{k}\partial_\xi^{n-k}K_y(0)(1-\mu)^ki^{n-k}M_y^{(k)}-M_y^{(n)} \\
	&= -\bigl(1-K_y(0)(1-\mu)^n\bigr)M_y^{(n)}+\sum_{k=0}^{n-1}\binom{n}{k}\partial_\xi^{n-k}K_y(0)(1-\mu)^ki^{n-k}M_y^{(k)}.
\end{align*}
From~\eqref{eq:K_y} we see that $K_y(0)=1$ for all $y\in [0,\,1]$. Moreover, for $k<n$ it results $i^{n-k}\partial_\xi^{n-k}K_y(\xi)\propto(i\nu)^{n-k}\partial_\xi^{n-k}\hat{C}(\nu\xi)$, thus $i^{n-k}\partial_\xi^{n-k}K_y(0)\propto\nu^{n-k}M_C^{(n-k)}$, where $M_C^{(j)}$ is the $j$-th moment of the connectivity distribution $C$. The proportionality constant is
$$ \cS_y:=
	\begin{cases}
		1-G(m_X^\infty) & \text{if } y<1-\beta(1) \\
		1 & \text{if } y\geq 1-\beta(1).
	\end{cases} $$
On the whole,
\begin{equation}
	\frac{dM_y^{(n)}}{dt}=-\bigl(1-(1-\mu)^n\bigr)M_y^{(n)}+\cS_y\sum_{k=0}^{n-1}\binom{n}{k}\nu^{n-k}(1-\mu)^kM_C^{(n-k)}M_y^{(k)}
	\label{eq:M_y^n}
\end{equation}
provides a triangular system of equations for the time evolution of all statistical moments of $p_y$. Thanks to~\eqref{eq:M_y^n} we can prove:
\begin{theorem}
Assume $p_y^0(v)=p_y(v,0)$ has finite moments of any order, i.e.
$$ \int_0^{+\infty}v^mp_y^0(v)\,dv<+\infty, \quad \forall\,m\in\N. $$
Then $p_y$ develops a fat tail if and only if the connectivity distribution $C$ is fat-tailed. In this case, the Pareto index of $p_y$ is the same as that of $C$.
\label{theo:tail}
\end{theorem}
\begin{proof}
Let $C$ be slim-tailed, so that $M_C^{(m)}<+\infty$ for all $m\in\N$. Assume then, by induction, that the first $n$ moments of $p_y$ are uniformly bounded in time, i.e. that there exist constants $\cM_j>0$ such that $M_y^{(j)}(t)\leq\cM_j$ for all $t>0$, $j=0,\,\dots,\,n$. We show that also the $(n+1)$-th moment of $p_y$ is uniformly bounded in time. Indeed, owing to~\eqref{eq:M_y^n}, the equation for $M_y^{(n+1)}$ turns out to be:
$$ \frac{dM_y^{(n+1)}}{dt}=-\bigl(1-(1-\mu)^{n+1}\bigr)M_y^{(n+1)}+\cS_y\sum_{k=0}^{n}\binom{n}{k}\nu^{n+1-k}(1-\mu)^kM_C^{(n+1-k)}M_y^{(k)}, $$
whence, multiplying both sides by $e^{\left(1-(1-\mu)^{n+1}\right)t}$ and integrating over $[0,\,t]$, $t>0$, we get
\begin{align*}
	M_y^{(n+1)}(t) &= e^{-\left(1-(1-\mu)^{n+1}\right)t}M_y^{(n+1)}(0) \\
	&\phantom{=} +\cS_y\sum_{k=0}^{n}\binom{n}{k}\nu^{n+1-k}(1-\mu)^kM_C^{(n+1-k)}\int_0^te^{-\left(1-(1-\mu)^{n+1}\right)(t-s)}M_y^{(k)}(s)\,ds.
\end{align*}
Therefore, since $1-(1-\mu)^{n+1}>0$,
\begin{align*}
	M_y^{(n+1)}(t) &\leq M_y^{(n+1)}(0) \\
	&\phantom{\leq} +\cS_y\sum_{k=0}^{n}\binom{n}{k}\nu^{n+1-k}(1-\mu)^kM_C^{(n+1-k)}\cM_k\int_0^te^{-\left(1-(1-\mu)^{n+1}\right)(t-s)}\,ds \\
	&= M_y^{(n+1)}(0) \\
	&\phantom{=} +\frac{\cS_y}{1-(1-\mu)^{n+1}}\sum_{k=0}^{n}\binom{n}{k}\nu^{n+1-k}(1-\mu)^kM_C^{(n+1-k)}\cM_k\left(1-e^{-\left(1-(1-\mu)^{n+1}\right)t}\right) \\
	&\leq M_y^{(n+1)}(0)+\frac{\cS_y}{1-(1-\mu)^{n+1}}\sum_{k=0}^{n}\binom{n}{k}\nu^{n+1-k}(1-\mu)^kM_C^{(n+1-k)}\cM_k=:\cM_{n+1}
\end{align*}
with $\cM_{n+1}<+\infty$ because of the boundedness of $M_y^{(n+1)}(0)$ and of all moments of $C$. Since the inductive assumption is clearly met for $n=0$, because $M_y^{(0)}(t)=\int_{\R_+}p_y(v,t)\,dv=1$ for all $t>0$, we conclude that all moments of $p_y$ of any order are uniformly bounded in time, hence that $p_y$ is slim-tailed.

Conversely, let $C$ be fat-tailed with Pareto index $r>0$, hence $M_C^{(j)}<+\infty$ for $j<r$ while $M_C^{(j)}=+\infty$ for $j\geq r$. From~\eqref{eq:M_y^n} we see that if $n<r$ then, at the right-hand side, $M_C^{(n-k)}<+\infty$ for all $k=0,\,\dots,\,n-1$. Therefore, arguing like before, we conclude that $M_y^{(n)}$ is uniformly bounded in time. On the contrary, if $n\geq r$ then $M_C^{(n-k)}=+\infty$ for all $k=0,\,\dots,\, n-r$ and consequently $M_y^{(n)}=+\infty$. Hence $p_y$ develops a fat tail with Pareto index $r$.
\end{proof}

\subsection{Numerical tests}
In this section, we provide numerical evidence of the result of Theorem~\ref{theo:tail}, namely that the connectivity distribution $C$ drives the formation of either fat- or slim-tailed popularity distributions $p_y$.

As a prototype of a fat-tailed probability distribution, we consider for $C$ an inverse-gamma distribution of parameters $a,\,b>0$, say $C\sim\operatorname{Inv-Gamma}(a,\,b)$, i.e.:
$$ C(c)=\frac{b^a}{\Gamma(a)}\cdot\frac{e^{-\frac{b}{c}}}{c^{a+1}}, \qquad c>0, $$
where $\Gamma(\cdot)$ denotes the gamma function. Since $C(c)\sim\frac{b^a}{\Gamma(a)}c^{-(a+1)}$ when $c\to +\infty$, this distribution features a fat tail with Pareto index $a$. We notice that a fat-tailed connectivity distribution may model the presence of \textit{influencers} in the social networks, i.e. users who, with non-negligible probability, may have a considerably large number of connections.

\begin{figure}[!t]
    \centering
    \subfigure[]{\includegraphics[width=.45\textwidth]{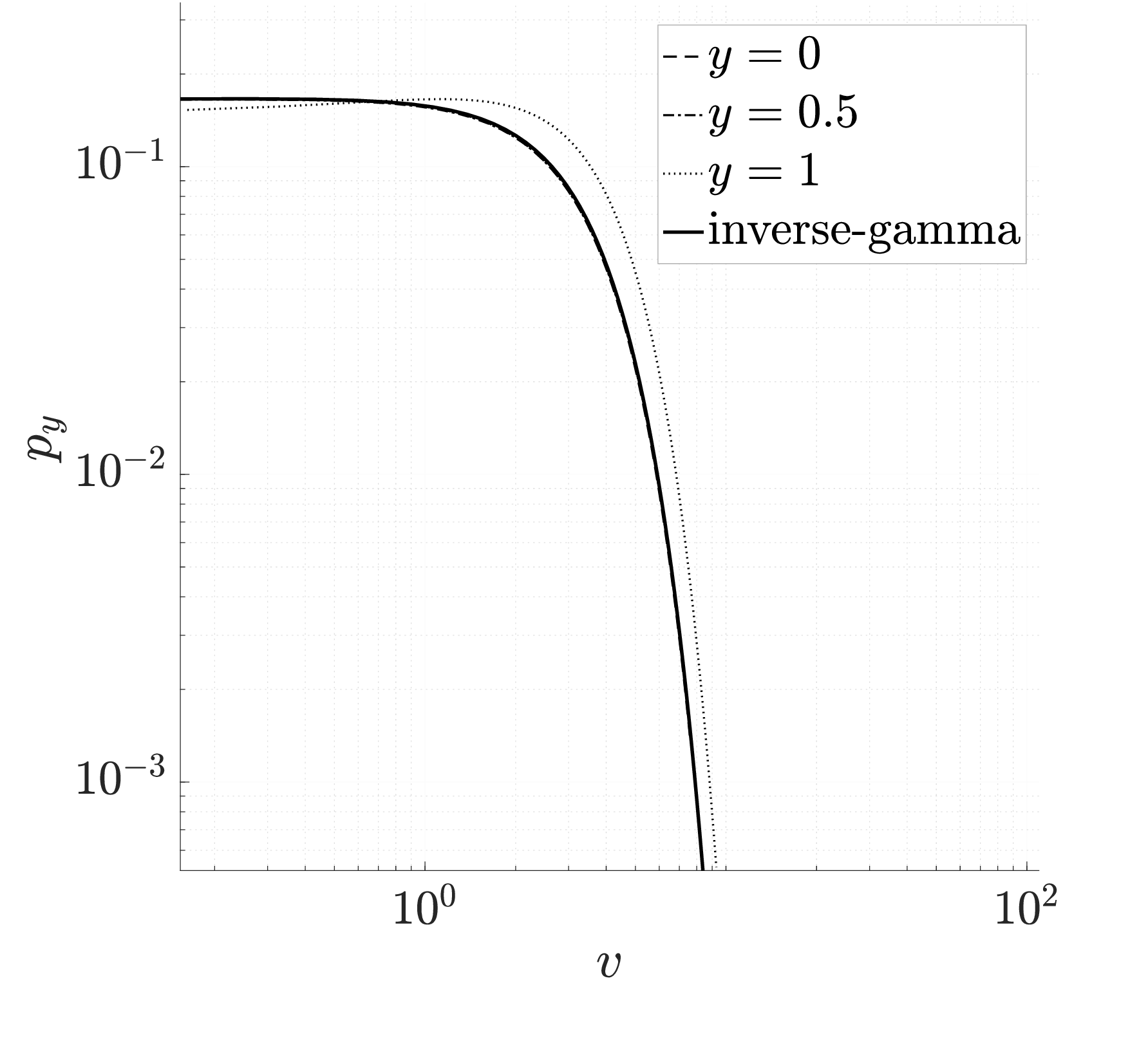} \label{fig:invgamma_tail-a1}}
    \qquad
    \subfigure[]{\includegraphics[width=.45\textwidth]{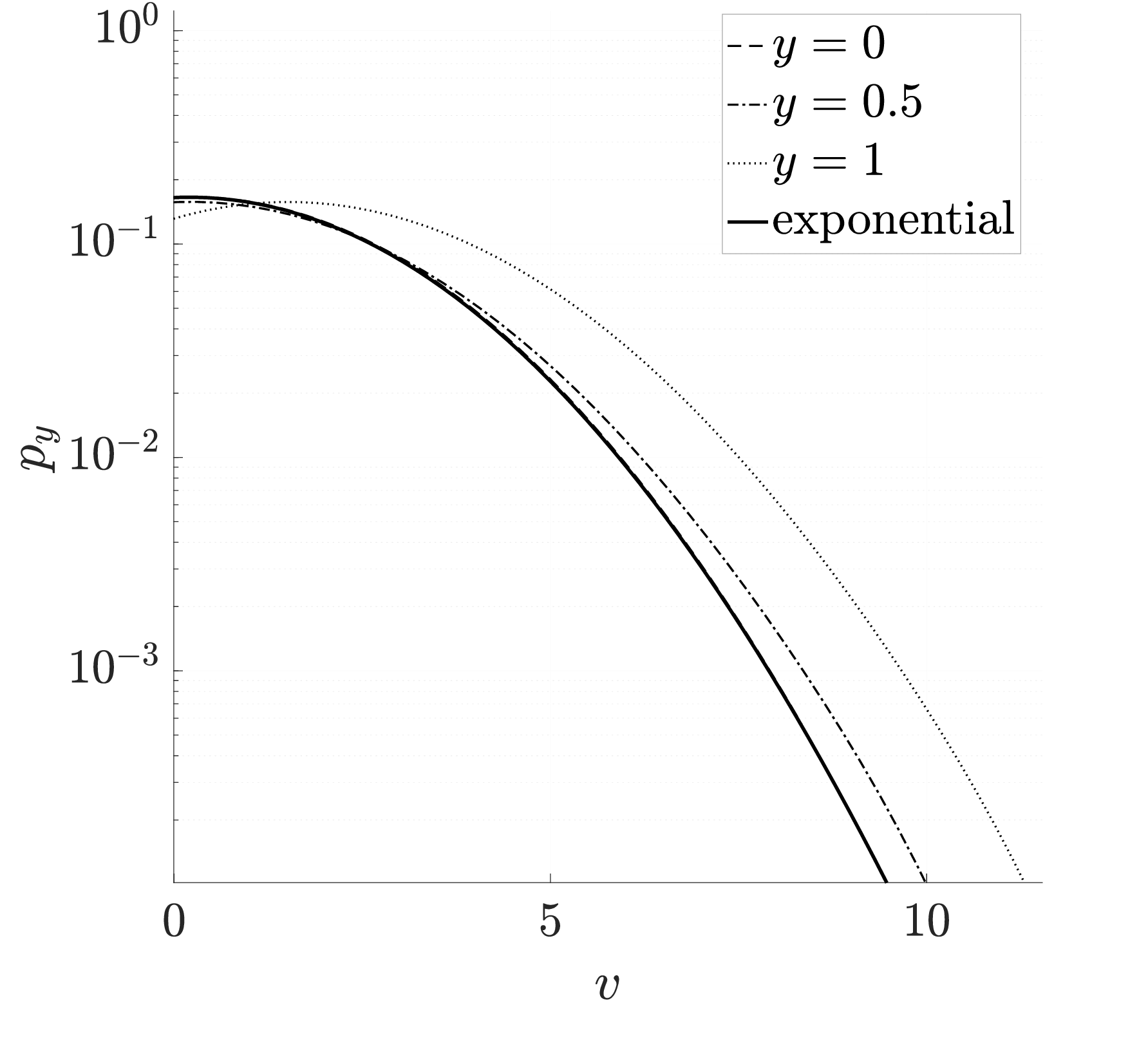} \label{fig:exp_tail-a1}}
    \caption{Popularity tails for $\alpha=1$ and: (a) $C\sim\operatorname{Inv-Gamma}(10,\,0.5)$, log scale; (b) $C\sim\operatorname{Exp}(5)$, log-linear scale}
    \label{fig:popularity_tails-a1}
\end{figure}

In Figure~\ref{fig:invgamma_tail-a1} we report some numerical solutions of~\eqref{eq:Boltzmann.p_y} with interaction rules~\eqref{eq:v'}-\eqref{eq:P} and $\beta\equiv 0.2$ constant obtained with a Monte Carlo particle algorithm. We consider, in particular, three levels of increasing reliability of news, $y=0,\,0.5,\,1$. In all cases, we observe clearly in log scale that the tail forming in the distribution $p_y$ follows the profile of the tail of $C$, hence it is fat with the same Pareto index as $C$, as expected from Theorem~\ref{theo:tail}. We also notice that the higher the reliability of news the higher the popularity that such news tends to gain. In particular, the popularity profile of non-completely reliable news ($y<1$) nearly coincides with that of the connections.

Conversely, as a prototype of a slim-tailed probability distribution we consider for $C$ an exponential distribution with parameter $a>0$, i.e. $C\sim\operatorname{Exp}(a)$,
$$ C(c)=ae^{-ac}, \qquad c\geq 0. $$
This distribution depicts a scenario of substantial absence of influencers in the social network, as the probability that a user has a large number of connections gets rapidly negligible.

With the same parameters and values of reliability of news as before, we show in Figure~\ref{fig:exp_tail-a1} the Monte Carlo numerical solution of~\eqref{eq:Boltzmann.p_y}. The slim (exponential) tail of $p_y$ is clearly visible by direct comparison with that of $C$ in log-linear scale, consistently with Theorem~\ref{theo:tail}. Moreover, also in this case we observe that the more reliable the news the higher the popularity it tends to gain and that completely unreliable news ($y=0$) develops a popularity profile which sticks closely to that of the connectivity distribution.

\begin{figure}[!t]
    \centering
    \subfigure[]{\includegraphics[width=.4\textwidth]{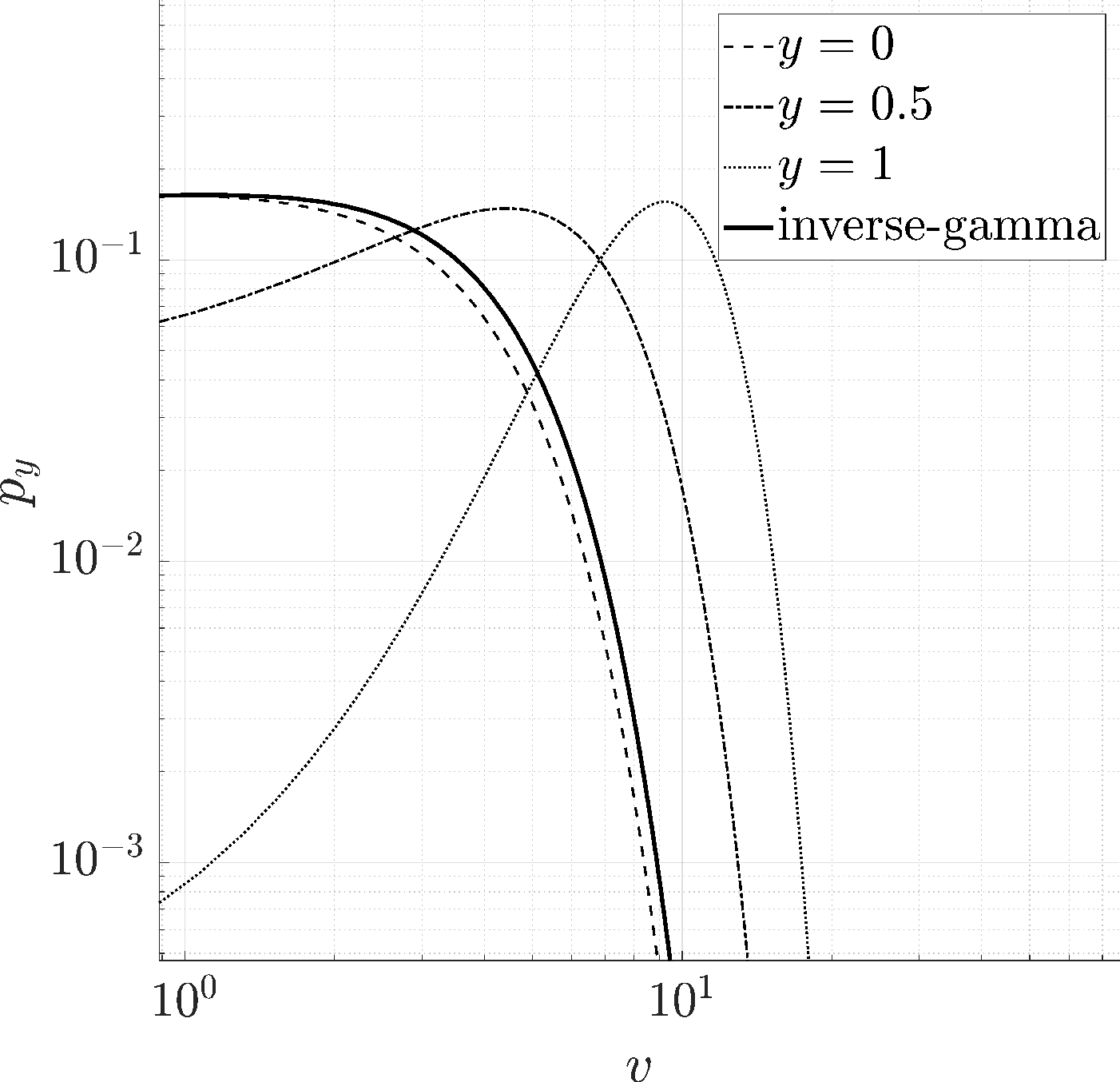} \label{fig:invgamma_tail}}
    \qquad
    \subfigure[]{\includegraphics[width=.4\textwidth]{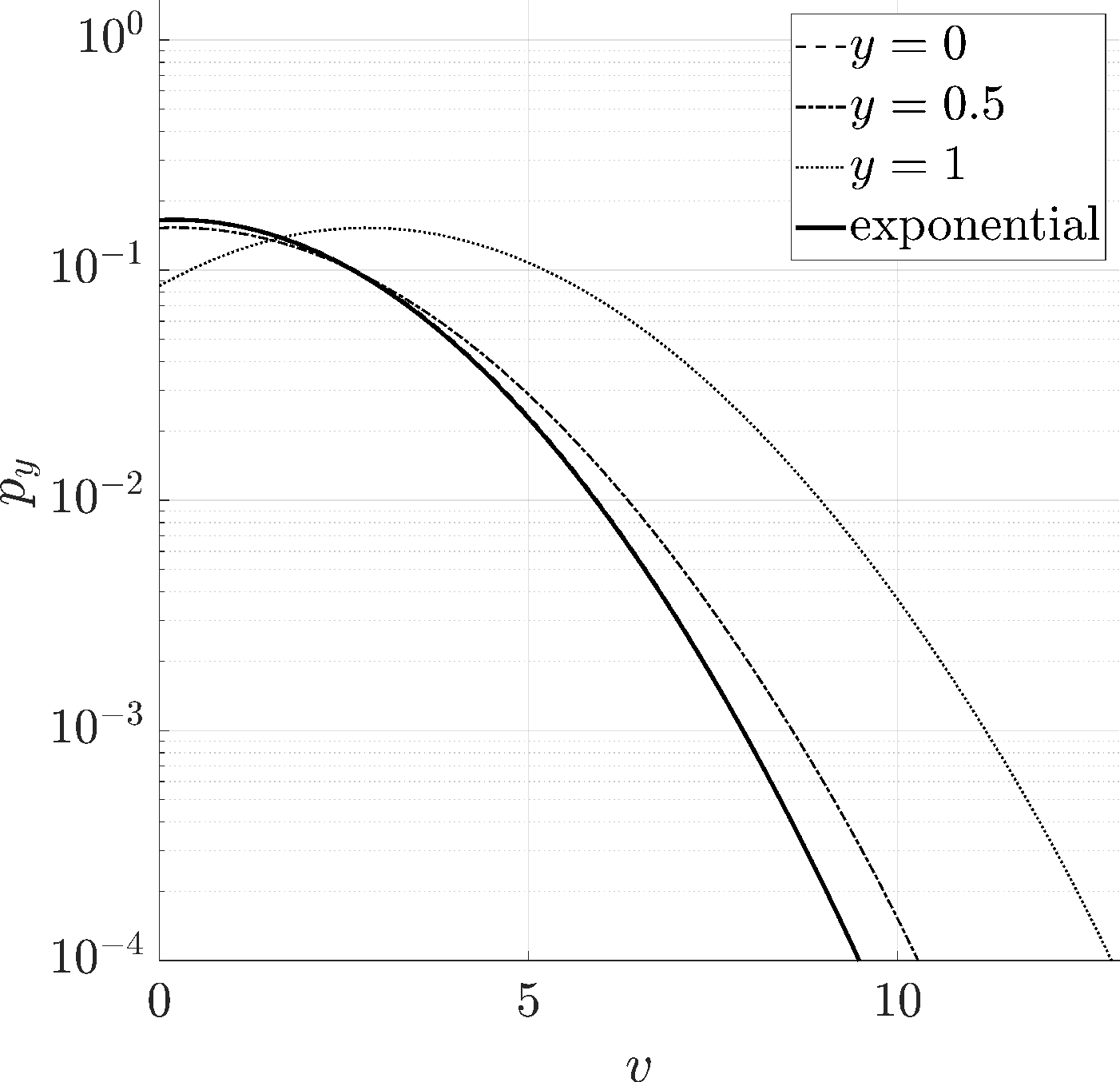} \label{fig:exp_tail}}
    \caption{Popularity tails for $\alpha<1$ fulfilling condition~\ref{ass:alpha} and: (a) $C\sim\operatorname{Inv-Gamma}(10,\,0.5)$, log scale; (b) $C\sim\operatorname{Exp}(5)$, log-linear scale}
    \label{fig:popularity_tails}
\end{figure}

Although analytically we considered only the case $\alpha=1$, the observations above still hold in the more general case $\alpha<1$ with condition~\eqref{ass:alpha} fulfilled. For instance, in Figure~\ref{fig:popularity_tails} we show the numerical results for the case $\alpha=0.8$, whence it is evident that the slim tail scenario virtually coincides with the corresponding one of the case $\alpha=1$, cf. Figures~\ref{fig:exp_tail-a1},~\ref{fig:exp_tail}. On the other hand, in the fat tail scenario, cf. Figure~\ref{fig:invgamma_tail}, the tail of the asymptotic popularity distribution still follows that of the connectivity distribution, but for small values of the popularity $v$, thus far from the tail, the trend is different compared to that displayed in Figure~\ref{fig:invgamma_tail-a1}.

\section{Conclusions}
\label{sect:conclusions}
In this paper, we have proposed two kinetic models describing, on one hand, a learning process leading individuals to build personal awareness about the reliability of news; and, on the other hand, the spread of fake news fostered by social media on the basis of the awareness of their users. We have formulated both models in terms of linear inelastic Boltzmann-type equations, proving in each case the existence and uniqueness of solutions and their trends towards equilibrium distributions, which we have been able to identify and characterise. Remarkably, all results refer to the genuinely collisional kinetic models without resorting to any limit description in special regimes of the model parameters.

Concerning the first model, we have assumed that the individuals increase or decrease their awareness against fake news depending on whether they are faced with news that they can or cannot detect as possibly partly unreliable. This model requires the distribution of the reliability of news as a fundamental input, a quantity that we have introduced mathematically in the abstract but that we have also indicated how to possibly infer from real data. We have shown that, despite the relative simplicity of the individual learning dynamics, the model possesses non-trivial equilibrium distributions in the form of clusters of awareness, which, under suitable assumptions, emerge in the long run regardless of the initial awareness distribution. This suggests a natural tendency of human societies to compartmentalise in awareness classes, whose number and distribution are controlled essentially by the rate of increase or decrease of individual awareness in the learning process. As a by-product, such a result supports the idea, widely used in the literature, that the spread of rumours in human societies might be described macroscopically by compartmental models inspired by the epidemiological ones.

Concerning instead the second model, we have assumed that the popularity that contents gain on social media depends on the ability of the users to ascertain their reliability and to decide, on such a basis, whether to repost them or not. In case of reposting, we have assumed furthermore that the effective penetration of a content is affected by the number of social connections of the reposting user. This second model requires as an input the awareness distribution returned by the previous model. As a simplifying working hypothesis, we have assumed that we could use the equilibrium awareness distribution, so as to decouple dynamically the two models. Such an assumption corresponds to the idea that the process of awareness formation acts as a background of the interactions among the social media users, being much quicker and persistent. More realistically, the two kinetic models should be coupled, so that the awareness formation depends on the contents a user interacts with on social media and that the content's popularity is simultaneously affected by the evolving awareness of the users. In our case, we have proved that the distribution of the popularity evolves towards a unique steady profile parameterised by the content's reliability, regardless of the initial popularity distribution. Moreover, we have characterised univocally such a steady distribution in terms of its Fourier transform and of its lower order statistical moments, namely the mean, energy, and variance. Out of them, we have shown that, on average, the popularity reached by sufficiently reliable news is invariably larger than that of poorly reliable news but also that partly fake contents reach typically the same large popularity as trustworthy ones. We have also proved that the tail of the popularity distribution is fully determined by that of the connectivity distribution, in particular that slim-/fat-tailed connectivity distributions entail slim-/fat-tailed popularity distributions, respectively, a reasonable result matching well with the intuitive expectation.

In this work, we have assumed for simplicity that the distribution of the connections of the users of the social network is fixed, viz. time-independent. A further direction of future research could involve considering \textit{co-evolving networks}, i.e. networks whose topology changes in time parallelly to the evolution of the popularity of news. This amounts to understanding the distribution of the connections as time-dependent, so as to capture the natural evolution of the social interactions over time.

\section*{Acknowledgements}
The authors gratefully acknowledge support from the Italian Ministry of University and Research (MUR) through the PRIN 2020 project No. 2020JLWP23 ``Integrated Mathematical Approaches to Socio-Epidemiological Dynamics'' and through the grant PRIN2022-PNRR project (No. P2022Z7ZAJ) “A Unitary Mathematical Framework for Modelling Muscular Dystrophies” (CUP: E53D23018070001) funded by the European Union - Next Generation EU.

The authors are members of GNFM (Gruppo Nazionale per la Fisica Matematica) of INdAM (Istituto Nazionale di Alta Matematica), Italy.

\appendix

\section{Existence and uniqueness of the solution to~\texorpdfstring{\eqref{eq:Boltzmann.x}}{}}
\label{sect:proof}

In this appendix, we prove the existence and uniqueness of the solution $f$ to~\eqref{eq:Boltzmann.x} with a prescribed initial condition supported in $[0,\,1]$. Notice that, owing to Proposition~\ref{prop:x'}, if $\supp{f^0}\subseteq [0,\,1]$ then $\supp{f(t)}\subseteq [0,\,1]$ for all $t>0$. Therefore, under the assumption of an initial datum supported in $[0,\,1]$, we may extend $f(t)$ to a probability measure defined on the whole $\R$ for all $t>0$, which vanishes outside the interval $[0,\,1]$. This is useful to identify a proper functional space of time-evolving probability measures on $\R$, designed to be complete with a metric built on the Fourier distance $d_s$, where to look for the solution of~\eqref{eq:Boltzmann.x} by means of the Banach fixed-point theorem.

For $s>0$, let
$$ \cP_s(\R):=\left\{\mu\in\cP(\R)\,:\,\int_{\R}\abs{x}^s\,d\mu(x)<+\infty\right\}. $$
Moreover, for $\gamma,\,C_{s+\gamma}>0$ let $\cP_{s,\gamma,C_{s+\gamma}}(\R)$ be the subset of $\cP_{s+\gamma}(\R)$ of probability measures with prescribed moments up to the order $[s]$ (the integer part of $s$) and such that $\int_{\R}\abs{x}^{s+\gamma}\,d\mu(x)\leq C_{s+\gamma}$, where $C_{s+\gamma}$ is independent of $\mu$. Then, based on~\cite[Proposition~2.6]{carrillo2007RMUP}, $\cP_{s,\gamma,C_{s+\gamma}}(\R)$ endowed with the metric $d_s$ is complete. 

If, for $0<s<1$, we choose $\gamma=1-s>0$ we obtain that the space $\cP_{s,1-s,C_1}(\R)$ of the probability measures in $\cP_1(\R)$ such that $\int_{\R}\abs{x}\,d\mu(x)$ is $\mu$-uniformly bounded by a constant $C_1>0$ is complete with the metric $d_s$. Notice that in $\cP_{s,1-s,C_1}(\R)$ the requirement of prescribed moments up to the order $[s]=0$ is satisfied straightforwardly.

The idea is to look for solutions of~\eqref{eq:Boltzmann.x} that at every $t>0$ belong to $\cP_{s,1-s,C_1}(\R)$. For this, however, we need to check preliminarily that solutions $f(t)$ of~\eqref{eq:Boltzmann.x} comply with the requirement of having $\int_{\R}\abs{x}f(x,t)\,dx$ uniformly bounded. Therefore, let $f=f(x,t)$ be any prospective solution to~\eqref{eq:Boltzmann.x}; setting $\varphi(x)=\abs{x}$ we obtain:
\begin{align*}
    \frac{d}{dt}\int_\R\abs{x}f(x,t)\,dx &= \int_\R\left[\int_0^{m_X}\left(\abs{x+\alpha(1-x)}-\abs{x}\right)g(y)\,dy\right. \\
    &\phantom{=} \left.+\int_{m_X}^1\left(\abs{x-\alpha x}-\abs{x}\right)g(y)\,dy\right]f(x,t)\,dx \\
    &= \int_\R\Bigl[\bigl(\abs{(1-\alpha)x+\alpha}-\abs{x}\bigr)G(m_X)-\alpha(1-G(m_X))\abs{x}\Bigr]f(x,t)\,dx \\
    &\leq \alpha\left(G(m_X)-\int_\R\abs{x}f(x,t)\,dx\right).
\end{align*}
Since $G(m_X)\leq 1$, this implies
$$ \int_\R\abs{x}f(x,t)\,dx\leq 1+\left(\int_\R\abs{x}f^0(x)\,dx-1\right)e^{-\alpha t}\leq 1, $$
where we have used the assumption $\supp{f^0}\subseteq [0,\,1]$ to get $\int_\R\abs{x}f^0(x)\,dx-1\leq 0$. Consequently, $\int_\R\abs{x}f(x,t)\,dx$ is uniformly bounded with respect to any prospective solution $f$ to~\eqref{eq:Boltzmann.x}; in particular, we can take $C_1=1$.

Owing to the completeness of $(\cP_{s,1-s,1}(\R),\,d_s)$, we have that also the space
$$ \pX:=C^0([0,\,T];\,\cP_{s,1-s,1}(\R)) $$
of time-continuous probability-measure-valued mappings equipped with the metric
$$ \varrho(f,g):=\sup_{t\in [0,\,T]}d_s(f(t),g(t)), \qquad f,\,g\in\pX $$
is complete for every $T>0$. In $(\pX,\,\varrho)$ we apply the Banach fixed-point theorem to prove:
\begin{theorem} \label{theo:exist_uniq}
Equation~\eqref{eq:Boltzmann.x} complemented with an initial condition $f^0\in\cP([0,\,1])$ admits a unique solution $f\in\pX$.
\end{theorem}
\begin{proof}
We begin by observing that, multiplying both sides by $e^t$,~\eqref{eq:Boltzmann.x} may be rewritten as
$$ \frac{d}{dt}\left(e^t\int_\R\varphi(x)f(x,t)\,dx\right)=e^t\int_0^1\int_\R\varphi(x')f(x,t)g(y)\,dx\,dy, $$
whence, integrating in time on $[0,\,t]$, $0<t\leq T$, and taking the initial condition into account,
\begin{equation}
    \int_\R\varphi(x)f(x,t)\,dx=e^{-t}\int_\R\varphi(x)f^0(x)\,dx+\int_0^te^{\tau-t}\int_0^1\int_\R\varphi(x')f(x,\tau)g(y)\,dx\,dy\,d\tau.
    \label{eq:eq_con_phi}
\end{equation}
The right-hand side may be regarded as the weak form of an operator $Q$ such that
\begin{align}
    \begin{aligned}[b]
        \int_\R\varphi(x)Q(f)(x,t)\,dx &:=
            e^{-t}\int_\R\varphi(x)f^0(x)\,dx \\
        &\phantom{:=} +\int_0^te^{\tau-t}\int_0^1\int_\R\varphi(x')f(x,\tau)g(y)\,dx\,dy\,d\tau
    \end{aligned}
    \label{eq:Q}
\end{align}
for every observable quantity $\varphi$. Hence~\eqref{eq:eq_con_phi} may be recast in the form
$$ \int_\R\varphi(x)f(x,t)\,dx=\int_\R\varphi(x)Q(f)(x,t)\,dx, $$
which, owing to the arbitrariness of $\varphi$, shows that the solution $f$ is a fixed point of $Q$.

To apply Banach fixed-point theorem we now show that $Q$ maps $\pX$ into itself and that it is a contraction on $\pX$.

To see that $Q$ maps $\pX$ into itself, i.e. $Q(\pX)\subseteq\pX$, we prove that $Q(f)(\cdot,t)$ is a probability measure and that the mapping $t\mapsto Q(f)(t)$ is continuous for every $f\in\pX$. The first property follows straightforwardly from~\eqref{eq:Q} by observing that
$$ \int_\R\varphi(x)Q(f)(x,t)\,dx\geq 0 $$
for every non-negative $\varphi$, which indicates that $Q(f)\geq 0$ for every $f\in\pX$, and that for $\varphi\equiv 1$ it results
\begin{align*}
    \int_\R Q(f)(x,t)\,dx &= e^{-t}\int_\R f^0(x)\,dx+\int_0^te^{\tau-t}\int_0^1\int_\R f(x,\tau)g(y)\,dx\,dy\,d\tau \\
    &= e^{-t}+e^{-t}(e^t -1)=1.
\end{align*}
The continuity in time requires instead to check that $d_s(Q(f)(t_0),Q(f)(t))$ vanishes when $t\to t_0$ for an arbitrary $t_0\in [0,\,T]$. To evaluate the Fourier distance $d_s$ we compute first
$$ \widehat{Q(f)}(\xi,t)=e^{-t}\hat{f}^0(\xi)+\int_0^te^{\tau-t}H(m_X,\xi)\hat{f}((1-\alpha)\xi,\tau)\,d\tau, $$
which may be obtained from~\eqref{eq:Q} with $\varphi(x)=e^{-i\xi x}$. Notice that the term $m_X=m_X(\tau)$ at the right-hand side, which is the solution to~\eqref{eq:mX}, does not depend on $f$ but only on $f^0$ because the solutions to~\eqref{eq:mX} are univocally determined by the initial condition $m_X^0=\int_0^1xf^0(x)\,dx$. After some manipulations we get
\begin{align*}
    d_s(Q(f)(t_0),Q(f)(t)) &\leq \abs{e^{-t}-e^{-t_0}}\frac{\abs{\hat{f}^0(\xi)-\hat{f}((1-\alpha)\xi,t_0)}}{\abs{\xi}^s} \\
    &\phantom{\leq} +\frac{\abs{e^{-i\alpha\xi}-1}}{\abs{\xi}^s}\cdot\abs*{e^{-t}
        \int_0^{t}e^{\tau}G(m_X)\,d\tau-e^{-t_0}\int_0^{t_0}e^{\tau}G(m_X)\,d\tau} \\
    &\phantom{\leq} +\abs{e^{-t}-e^{-t_0}}
        \int_0^{t_0}e^\tau\frac{\abs{\hat{f}((1-\alpha)\xi,\tau)-\hat{f}((1-\alpha)\xi,t_0)}}{\abs{\xi}^s}\,d\tau \\
    &\phantom{\leq} +e^{-t}\int_{t_0}^{t}e^\tau\frac{\abs{\hat{f}((1-\alpha)\xi,\tau)-\hat{f}((1-\alpha)\xi,t_0)}}{\abs{\xi}^s}\,d\tau \\
    &\leq \abs{e^{-t}-e^{-t_0}}(2^{1-s}\alpha^s+(1-\alpha)^sd_s(f^0,f(t_0)) \\
    &\phantom{\leq} +2^{1-s}\alpha^s\abs*{e^{-t}\int_0^{t}e^{\tau}G(m_X)\,d\tau-e^{-t_0}\int_0^{t_0}e^{\tau}G(m_X)\,d\tau} \\
    &\phantom{\leq} +(1-\alpha)^s\abs{e^{-t}-e^{-t_0}}\int_0^{t_0}e^\tau d_s(f(\tau),f(t_0))\,d\tau \\
    &\phantom{\leq} +(1-\alpha)^se^{-t}\int_{t_0}^{t}e^\tau d_s(f(\tau),f(t_0))\,d\tau,
\end{align*}
whence we see that the right-hand side vanishes in the limit $t\to t_0$ due to the continuity of the exponential and integral functions. For completeness, we record that we have used the following estimates:
$$ \frac{\abs{e^{-i\alpha\xi}-1}}{\abs{\xi}^s}\leq 2^{1-s}\alpha^s, \qquad \abs{H(m_X,\xi)}\leq 1 $$
from the proof of Theorem~\ref{theo:ds}, and
\begin{align*}
    \frac{\abs{\hat{f}^0(\xi)-\hat{f}((1-\alpha)\xi,t_0)}}{\abs{\xi}^s} &= \frac{\abs{\hat{f}^0(\xi)-\hat{f}^0((1-\alpha)\xi)
        +\hat{f}^0((1-\alpha)\xi)-\hat{f}((1-\alpha)\xi,t_0)}}{\abs{\xi}^s} \\
    &\leq\frac{\abs{\hat{f}^0(\xi)-\hat{f}^0((1-\alpha)\xi)}}{\abs{\xi}^s}
        +\frac{\abs{\hat{f}^0((1-\alpha)\xi)-\hat{f}((1-\alpha)\xi,t_0)}}{\abs{\xi}^s}
\end{align*}
with in particular
\begin{align*}
    \frac{\abs{\hat{f}^0(\xi)-\hat{f}^0((1-\alpha)\xi)}}{\abs{\xi}^s} &=
        \frac{1}{\abs{\xi}^s}\abs*{\int_0^1f^0(x)(e^{-i\xi x}-e^{-i(1-\alpha)\xi x})\,dx} \\
    &\leq \frac{1}{\abs{\xi}^s}\int_0^1f^0(x)\abs*{e^{-i\xi x}-e^{-i(1-\alpha)\xi x}}\,dx \\
    &= \int_0^1f^0(x)\abs{e^{-i\xi x}}\cdot\abs{x}^s\frac{\abs{1-e^{i\alpha\xi x}}}{\abs{\xi x}^s}\,dx \\
    &\leq 2^{1-s}\alpha^s\int_0^1f^0(x)\abs{x}^s\,dx \\
    &\leq 2^{1-s}\alpha^s.
\end{align*}
and
\begin{align*}
    \frac{\abs{\hat{f}^0((1-\alpha)\xi)-\hat{f}((1-\alpha)\xi,t_0)}}{\abs{\xi}^s} &=
        (1-\alpha)^s\frac{\abs{\hat{f}^0((1-\alpha)\xi)-\hat{f}((1-\alpha)\xi,t_0)}}{\abs{(1-\alpha)\xi}^s} & \text{(set $\eta:=(1-\alpha)\xi$)} \\
    &\leq (1-\alpha)^s\sup_{\eta\in\R\setminus\{0\}}\frac{\abs{\hat{f}^0(\eta)-\hat{f}(\eta,t_0)}}{\abs{\eta}^s} \\
    &= (1-\alpha)^sd_s(f^0,f(t_0)).
\end{align*}

Finally, to see that $Q$ is a contraction on $\pX$ we consider $f_1,\,f_2\in\pX$ and study
\begin{align*}
    \frac{\abs{\widehat{Q(f_2)}(\xi,t)-\widehat{Q(f_1)}(\xi,t)}}{\abs{\xi}^s} &\leq
        \int_0^te^{\tau-t}\abs{H(m_X,\xi)}\frac{\abs{\hat{f}_2((1-\alpha)\xi,\tau)-\hat{f}_1((1-\alpha)\xi,\tau)}}{\abs{\xi}^s}\,d\tau \\
    &\leq (1-\alpha)^s\int_0^te^{\tau-t}d_s(f_1(\tau),f_2(\tau))\,d\tau \\
    &\leq (1-\alpha)^s(1-e^{-t})\rho(f_1,f_2),
\end{align*}
therefore
$$ \varrho(Q(f_1),Q(f_2))\leq (1-\alpha)^s(1-e^{-T})\varrho(f_1,f_2). $$
Since $0\leq\alpha<1$ and $T>0$, we deduce that $(1-\alpha)^s(1-e^{-T})<1$ for all $T>0$, hence that $Q$ is a contraction on $\pX$ for every $T>0$.
\end{proof}
	
\bibliographystyle{plain}
\bibliography{FmLnTa-fake_news}
\end{document}